\documentclass[12pt]{article}
\usepackage{algorithm, algpseudocode}
\usepackage{amsmath, amsopn, amsthm}
\usepackage{graphicx,psfrag,epstopdf, subfig}
\usepackage{enumerate}
\usepackage{url} 
\usepackage{tikz}
  \usetikzlibrary{shapes,positioning,arrows}

\newcommand{\blind}{0}

\newcommand{\BALD}{\begin{aligned}}
\newcommand{\EALD}{\end{aligned}}
\newcommand{\BALDS}{\begin{aligned*}}
\newcommand{\EALDS}{\end{aligned*}}
\newcommand{\BCAS}{\begin{cases}}
\newcommand{\ECAS}{\end{cases}}
\newcommand{\BEAS}{\begin{eqnarray*}}
\newcommand{\EEAS}{\end{eqnarray*}}
\newcommand{\BEQ}{\begin{equation}}
\newcommand{\EEQ}{\end{equation}}
\newcommand{\BIT}{\begin{itemize}}
\newcommand{\EIT}{\end{itemize}}
\newcommand{\BMAT}{\begin{bmatrix}}
\newcommand{\EMAT}{\end{bmatrix}}
\newcommand{\BNUM}{\begin{enumerate}}
\newcommand{\ENUM}{\end{enumerate}}
\newcommand{\eg}{e.g.}

\newcommand{\ie}{i.e.}

\newcommand{\BA}{\begin{array}}
\newcommand{\EA}{\end{array}}

\newcommand{\reals}{\mathbf{R}}

\DeclareMathOperator*{\maximize}{maximize}
\DeclareMathOperator*{\minimize}{minimize}

\renewcommand{\Pr}{\mathbf{P}}

\newcommand{\subjectto}{\textrm{subject to}}

\DeclareMathOperator{\linspan}{span}
\newcommand{\norm}[1]{\left\| #1 \right\|}

\newcommand{\iid}{\emph{i.i.d.}}

\DeclareMathOperator{\cone}{cone}
\DeclareMathOperator{\conv}{conv}

\DeclareMathOperator{\ext}{ext}

\newcommand{\bS}{\mathbf{S}}

\newcommand{\cN}{\mathcal{N}}

\newcommand{\Cap}{\mathrm{Cap}}

\DeclareMathOperator{\diam}{diam}
\DeclareMathOperator{\unif}{unif}

\DeclareMathOperator{\arccot}{arccot}

\numberwithin{equation}{section}
  \theoremstyle{plain}

\newtheorem{theorem}{Theorem}[section]
\newtheorem{assumption}[theorem]{Assumption}
\newtheorem{corollary}[theorem]{Corollary}

\newtheorem{lemma}[theorem]{Lemma}

\newtheorem{remark}[theorem]{Remark}

\begin{document}

\def\spacingset#1{\renewcommand{\baselinestretch}%
{#1}\small\normalsize} \spacingset{1}


\if0\blind
{
  \title{\bf A geometric approach to archetypal analysis and non-negative matrix factorization}
  \author{Anil Damle\thanks{
    The authors gratefully acknowledge Trevor Hastie, Jason Lee, Philip Pauerstein, Michael Saunders, Jonathan Taylor, Jennifer Tsai and Lexing Ying for their insightful comments. Trevor Hastie suggested the group-lasso approach to selecting extreme points. A. Damle was supported by a NSF Graduate Research Fellowship DGE-1147470. Y. Sun was partially support by the NIH, grant U01GM102098.}\hspace{.2cm}\\
    Institute for Computational and Mathematical Engineering, Stanford University\\
    and \\
    Yuekai Sun \\
    Department of Statistics, UC Berkeley}
  \maketitle
} \fi

\if1\blind
{
  \bigskip
  \bigskip
  \bigskip
  \begin{center}
    {\LARGE\bf A geometric approach to archetypal analysis and non-negative matrix factorization}
\end{center}
  \medskip
} \fi

\bigskip
\begin{abstract}
Archetypal analysis and non-negative matrix factorization (NMF) are staples in a statisticians toolbox for dimension reduction and exploratory data analysis. We describe a geometric approach to both NMF and archetypal analysis by interpreting both problems as finding extreme points of the data cloud. We also develop and analyze an efficient approach to finding extreme points in high dimensions. For modern massive datasets that are too large to fit on a single machine and must be stored in a distributed setting, our approach makes only a small number of passes over the data. In fact, it is possible to obtain the NMF or perform archetypal analysis with just two passes over the data.
\end{abstract}

\noindent%
\vfill
\hfill {\tiny technometrics tex template (do not remove)}

\newpage
\spacingset{1.45} 
\section{Introduction}
\label{sec:introduction}

Archetypal analysis (by \cite{cutler1994archetypal}) and non-negative matrix factorization (NMF) (by \cite{lee1999learning}) are staple approaches to finding low-dimensional structure in high-dimensional data. At a high level, the goal of both tasks boil down to approximating a data matrix $X\in\reals^{n\times p}$ with factors $W\in\reals^{n\times k}$ and $H\in\reals^{k\times p}:$
\BEQ
X\approx WH.
\label{eq:matrix-nearness-problem}
\EEQ

In archetypal analysis, the rows of $H$ are archetypes, and the rows of $W$ are convex combinations that (approximately) represent the data points. The archetypes are forced to be convex combinations of the data points:
\BEQ
H = BX
\label{eq:archetype-constraint}
\EEQ
By requiring the data points to be convex combinations of the rows of $H$, archetypal analysis forces the archetypes to lie on the convex hull of the data cloud. Thus the archetypes are interpretable as ``pure'' data points. Given \eqref{eq:matrix-nearness-problem} and \eqref{eq:archetype-constraint}, a natural approach to archetype analysis is to solve the optimization problem
\BEQ
\minimize_{W,B}\frac12\norm{X - WBX}_F^2.
\label{eq:archetypal-analysis-problem}
\EEQ
The problem is solved by alternating minimization over $W$ and $B.$ The overall problem is non-convex, so the algorithm converges to a local minimum of the problem.

In NMF, the entries of $X$ and $H$ are also required to be non-negative. NMF is usually motivated as an alternative to principal components analysis (PCA), in which the data and components are assumed to be non-negative. In some scientific applications, requiring the components to be non-negative makes the factorization consistent with physical reality, and gives more interpretable results versus more classical tools. Given its many applications NMF has been studied extensively, and many clever heuristics were proposed over the years to find NMFs. \cite{lee2001algorithms} proposes a multiplicative update algorithm that solves the optimization problem
\BEQ
\maximize_{W,H} \sum_{i=1}^n\sum_{j=1}^p x_{ij}\log(WH)_{ij} - (WH)_{ij}.
\label{eq:multiplicative-algorithm}
\EEQ
The solution to \eqref{eq:multiplicative-algorithm} is the maximum likelihood estimator for a model in which $x_{ij}$ is Poisson distributed with mean $(WH)_{ij}.$ An alternative approach is to minimize the residual sum-of-squares $\frac12\norm{X - WH}_F^2$ by alternating minimization over $W$ and $H.$ Although these heuristics often perform admirably, none are sure to return the correct factorization.\footnote{\cite{vavasis2009complexity} showed computing the NMF is, in general, NP-hard.} However, a recent line of work started by \cite{arora2012computing} showed that the problem admits an efficient solution when the matrix is \emph{separable}. In this work, we also focus on factorizing separable matrices.

Although the goal of both archetypal analysis and NMF boil down to the same matrix nearness problem, the two approachs are usually applied in different settings and have somewhat different goals. In NMF, we require $k \le p.$ Otherwise, we may obtain a trival exact NMF by setting $W = I$ and $H = X.$ In archtypal analysis, we require $k \le n,$ but allow $k \ge p.$ The archetype constraint \eqref{eq:archetype-constraint} implies the archetypal approximation will not be perfect even if we allow $k \ge p.$

\subsection{Separable archetypal analysis and NMF}

The notion of a separable NMF was introduced by \cite{donoho2003when} in the context of image segmentation.

\begin{assumption}
A non-negative matrix $X\in\reals^{n\times p}$ is separable if and only if there exists a permutation matrix $P\in\reals^{n\times n}$ such that
$$
PX = \BMAT I \\ W_2\EMAT H,
$$
where $W_2\in\reals^{(n-k)\times k}$ and $H\in\reals^{k\times p}$ are non-negative.
\end{assumption}

The notion of separability has a geometric interpretation. It asserts that the conical hull of a small subset of the data points (the points that form $H$) contain the rest of the data points, \ie{} the rows of $X$ are contained in the cone generated by the rows in $H$:
$$
\{x_1,\dots,x_n\} \subset \cone(\{h_1,\dots,h_k\}).
$$
The rows of $H$ are the extreme rays of the cone. If $x_1,\dots,x_n$ are normalized to lie on some (affine) hyperplane $A,$ then the separability assumption implies $x_1,\dots,x_n$ are contained in a polytope $P\subset A$ and $h_1,\dots,h_k$ are the \emph{extreme points} of $P$.

The separable assumption is justified in many applications of NMF; we give two common examples.
\BNUM
\item In hyperspectral imaging, a common post-processing step is \emph{unmixing}: detecting the materials in the image and estimating their relative abundances. Unmixing is equivalent to computing a NMF of the hyperspectral image. The separability assumption asserts for each material in the image, there exists at least one pixel containing only that material. The assumption is so common that it has a name: the \emph{pure-pixel assumption.}
\item In document modeling, documents are usually modeled as additive mixtures of topics. Given a collection of documents, the NMF of the \emph{document-term matrix} reveals the topics in the collection. The separability assumption is akin to assuming for each topic, there is a word that only appears in documents concerning that topic. Such special words are called \emph{anchor words.}
\ENUM

Given the geometric interpretation of separability, it is straightfoward to generalize the notion to archetypal analysis. In archetypal analysis, the archetypes $h_k,\dots,h_k$ are usually convex combinations of the data points. If we force the archetypes to be data points, \ie{} enforce
\[
H = EX,
\]
where the rows of $E\in\reals^{k\times n}$ are a subset of the rows of the identity matrix, then we are forcing the archetypes to be extreme points of the data cloud. The analogous optimization problem for separable archetype analysis is
\BEQ
\minimize_{W,E}\frac12\norm{X - WEX}_F^2,
\label{eq:separable-archetypal-analysis-problem}
\EEQ
where $E$ is constrained to consist of a subset of the rows of the identity. It seems \eqref{eq:separable-archetypal-analysis-problem} is harder than \eqref{eq:archetypal-analysis-problem} because minimizing over $E$ is a combinatorial problem. However, as we shall see, separability allows us to reduce archetypal analysis and NMF to an extreme point finding problem that admits an efficient solution.

\subsection{Related work on separable NMF}

To place our algorithm in the correct context, we review the recently proposed algorithms for computing a NMF when $X$ is separable. All these algorithms exploit the geometric interpretation of a separability and find the extreme points/rays of the smallest polytope/cone that contains the columns of $X$.
\BNUM
\item \cite{arora2012computing} describe a method which checks whether each column of $X$ is an extreme point by solving a linear program (LP). Although this is the first polynomial time algorithm for separable NMF, solving a LP per data point is not practical when the number of data points is large.
\item \cite{bittorf2012factoring} make the observation that $X$ has the form
\[
X = P^T\BMAT I & 0\,\\ W_2 & 0\, \EMAT X = CX,
\]
for some $C\in\reals^{n\times n}$. To find $C$, they solve a LP with $n^2$ variables. To handle large problems, they use a first-order method to solve the LP. \cite{gillis2013robustness} later developed a post-processing procedure to make the approach in \cite{bittorf2012factoring} more robust to noise.
\item \cite{esser2012convex} formulate the column subset selection problem as a dictionary learning problem and use $\ell_{1,\infty}$ norm regularization to promote sparse dictionaries. Although convex, the dictionary learning problem may not find the sparsest dictionary.
\item \cite{gillis2012fast} describe a family of recursive algorithms that maximize strongly convex functions over the cloud of points to find extreme points. Their algorithms are based on the intuition that the maximum of a strongly convex function over a polytope is attained at an extreme point.
\item \cite{kumar2013fast} describe an algorithm called \textsc{Xray} for finding the extreme rays by ``expanding'' a cone one extreme ray at a time until all the columns of $X$ are contained in this cone.
\ENUM
Algorithms 1, 2, and 3 require the solution of convex optimization problems and are not suited to factorizing large matrices (\eg{} document-term matrices where $n\sim 10^9$). Algorithms 1, 2, and 5 also require the non-negative rank $k$ to be known \emph{a priori}, but $k$ is usually not known in practice. Algorithms 1 and 2 also depend heavily on separability, while our approach gives interpretable results even when the matrix is not separable. Finally, algorithm 4 requires $U$ to be full rank, but this may not be the case in practice.

The idea of finding the extreme points of a point cloud by repeatedly maximizing and minimizing linear functions is not new. An older algorithm for unmixing hyperspectral images is pure-pixel indexing (PPI) by \cite{boardman1994geometric}. PPI is a popular technique for unmixing due to its simplicity and availability in many image analysis packages. The geometric intuition behind PPI is the same as the intuition behind our algorithm, but there are few results concerning the performance of this simple algorithm. Since its introduction, many extensions and modifications of the core algorithm have been proposed; \eg{} \cite{nascimento2005vertex, chang2006fast}. Recently, \cite{ding2013topic} propose algorithms for topic modeling based on similar ideas.

\section{Archetype pursuit}
\label{sec:archetype-pursuit}

Given a cloud pf points in the form of a data matrix $X\in\reals^{n\times p},$ we focus on finding the extreme points of the cloud. We propose a randomized approach that finds all $k$ extreme points in $O(pk\log k)$ floating point operations (flops) with high probability. In archetypal analysis, the extreme points are the archetypes. Thus we refer to our approach as archetype pursuit. After finding the extreme points, we solve for the weights by non-negative least squares (NNLS):
\BEQ
\label{eq:nnls}
\begin{aligned}
& \minimize_W & & \frac12\norm{X-WH}_F^2 \\
& \subjectto & & W \ge 0.
\end{aligned}
\EEQ

The geometric intuition behind archetype pursuit is simple: the extrema of linear functions over a convex polytope is attained at extreme points of the polytope. By repeatedly maximizing and minimizing linear functions over the point cloud, we find the extreme points. As we shall see, by choosing \emph{random} linear functions, the number of optimizations required to find all the extreme points with high probability depends only on the number of extreme points (and not the total number of points).

Another consequence of the geometric interpretation is the observation that projecting the point cloud onto a random subspace of dimension at least $k+1$ preserves all of the extreme points with probability one. Such a random projection could be used as a precursor to existing NMF algorithms as it effectively reduces the dimension of the problem. However, given the nature of the algorithm we discuss here a random projection of this form would yield no additional benefits.

\subsection{A prototype algorithm}

We first describe and analyze a proto-algo\-rithm for finding the extreme points of a point cloud. This algorithm closely resembles the original PPI algorithm as described in \cite{boardman1994geometric}.

\begin{algorithm}
\caption{Proto-algorithm}
\label{alg:proto-algorithm}
\begin{algorithmic}[1]
\Require $X\in\reals^{n\times p}$
\State Generate an $p\times m$ random matrix $G$ with independent standard normal entries.
\State Form the product $XG$.
\State Find the indices of the max $I_{\max}$ and min $I_{\min}$ in each column of $XG$.
\State Return $H = X_{I_{\max}\cup I_{\min}}.$
\end{algorithmic}
\end{algorithm}

The proto-algorithm finds points attaining the maximum and minimum of random linear functions on the point cloud. Each column of the random matrix $G$ is a random linear function, hence forming $XG$ evaluates $m$ linear functions at the $n$ points in the cloud. A natural question to ask is how many optimizations of random linear functions are required to find all the extreme points with high probability?

\subsubsection{Relevant notions from convex geometry}
\label{sec:convex-geometry}

Before delving into the analysis of the proto-algorithm, we review some concepts from convex geometry that appear in our analysis. A convex \emph{cone} $K\subset\reals^p$ is a convex set that is positively homogeneous, \ie{} $K = \lambda K$ for any $\lambda \ge 0$. Two examples are \emph{subspaces} and \emph{the non-negative orthant} $\reals_+^p$. A cone is \emph{pointed} if it does not contain a subspace. A subspace is not a pointed cone, but the non-negative orthant is. The \emph{polar cone} $K^\circ$ of a cone $K$ is the set
\[
K^\circ := \{y\in\reals^p\mid x^Ty\le 0\text{ for any }x\in K\}.
\]
The notion of polarity is a generalization of the notion of orthogonality. In particular, the polar cone of a subspace is its orthogonal complement. Given a convex cone $K\subset\reals^p$, any point $x\in\reals^p$ has an orthogonal decomposition into its projections\footnote{Given a closed convex set $C\subset\reals^p$, the \emph{projection} of a point $x$ onto $C$ is simply the closest point to $x$ in $C$, \ie{}
\[
\norm{x - P_C(x)}_2 = \inf_y\{\norm{x-y}_2\mid y\in C\}.
\]} onto $K$ and $K^\circ$. Further, the components $P_K(x)$ and $P_{K^\circ}(x)$ are orthogonal. This implies a conic Pythagorean theorem, \ie
\BEQ
\norm{x}_2^2 = \norm{P_K(x)}_2^2 + \norm{P_{K^\circ}(x)}_2^2.
\label{eq:conic-pythagorean-theorem}
\EEQ

Two cones that arise in our analysis deserve special mention: normal and circular cones. The \emph{normal cone} of a convex set $C$ at a point $x$ is the cone
\[
N_C(x) = \{w\in\reals^p\mid w^T(y-x) \le 0\text{ for any }y\in C\}.
\]
It is so called because it comprises the (outward) normals of the supporting hyperplanes at $x$. The polar cone of the normal cone is the \emph{tangent cone}:
\[
T_C(x) = \cone\left(C -x \right).
\]
The tangent cone is a good local approximation to the set $C$. A \emph{circular cone} or \emph{ice cream cone} is a cone of the form
\[
K = \{x\in\reals^p\mid \theta^Tx\ge t\norm{x}_2\}\text{ for some }\theta\in\bS^{p-1},t \in (0,1].
\]
In other words, a circular cone is a set of points making an angle smaller than $\arccos(t)$ with the axis $a$ ($\arccos(t)$ is called the \emph{angle} of the cone). The polar cone of a circular cone (with axis $a\in\reals^p$ and aperture $\arccos(t)$) is another circular cone (with axis $-a$ and angle $\frac{\pi}{2} - \arccos(t)$).

A \emph{solid angle} is a generalization of the angles in the (Cartesian) plane to higher dimensions. Given a (convex) cone $K\subset\reals^p$, the \emph{solid angle} $\omega(K)$ is the proportion of space that the cone $K$ occupies; \ie{} if we pick a random direction $x\in\reals^p$, the probability that $x\in K$ is the solid angle at the apex of $K$. Mathematically, the solid angle of a cone $K$ is given by
\[
\omega(K) = \int_K e^{-\pi\norm{x}_2^2}dx,
\]
where the integral is taken over $\linspan(K)$. By integrating over the linear hull of $K$, we ensure $\omega(K)$ is an \emph{intrinsic} measure of the size of $K$. When $K$ is full-dimensional (\ie{} $\linspan(K) = \reals^p$), the solid angle is equivalent to (after a change of variables)
\begin{align}
\omega(K) &= \frac{1}{(2\pi)^{p/2}}\int_K e^{-\frac12\norm{x}_2^2}dx = \Pr(z\in K),\,z\sim\cN(0,I) \label{eq:gaussian-solid-angle} \\
&= \Pr(\theta\in K\cap\bS^{p-1}),\,\theta\sim\unif(\bS^{p-1}). \label{eq:spherical-solid-angle}
\end{align}

For a convex polytope $P\subset\reals^p$ (the convex hull of finitely many points), the solid angles of the normal cones at its extreme points also form a probability distribution over the extreme points, \ie{}
\[
\sum_{h_i\,\in\,\ext(P)}\omega(N_P(h_i)) = 1.
\]
Furthermore, $\omega(N_P(h_i)) \in \left[0,\frac12\right).$ Calculating the solid angle of all but the simplest cone in $\reals^p,\,p > 3$ is excruciating. Fortunately, we know bounds on solid angles for some cones.

For a point $\theta\in\bS^{p-1}$, the set
\[
\Cap\left(\theta,t\right) = \left\{v\in\bS^{p-1}\mid \theta^Tv \ge t\right\}
\]
is called a \emph{spherical cap} of height $t$. Since the solid angle of a (convex) cone $K\subset\reals^p$ is the proportion of $\bS^{p-1}$ occupied by $K$, the solid angle of a circular cone with angle $\arccos(t)$ is given by the normalized area of the spherical cap $\Cap\left(\theta,t\right)$ for any $\theta\in\bS^{p-1}$:
\[
\omega\left(\{x\in\reals^p\mid \theta^Tx\ge t\norm{x}_2\}\right) = \sigma_{p-1}\left(\Cap(\theta,t)\right),
\]
where $\sigma_{p-1}$ is the rotation-invariant measure on $\bS^{p-1}$ of total mass 1.

To state estimates for the area of spherical caps, it is sometimes convenient to measure the size of a cap in terms of its \emph{chordal radius}. The spherical cap of radius $r$ around a point $\theta\in\bS^{p-1}$ is
\[\textstyle
\left\{v\in\bS^{p-1}\mid \norm{\theta - v}_2\le r\right\} = \Cap\left(\theta,\frac12r^2-1\right).
\]
Two well-known estimates for the area of spherical caps are given in \cite{ball1997elementary}. The lower bound is exactly \cite[Lemma 2.3]{ball1997elementary}, and the upper bound is a sharper form of \cite[Lemma 2.2]{ball1997elementary}.

\begin{lemma}[Lower bound on the area of spherical caps]
\label{lem:area-spherical-cap-2}
The spherical cap of radius $r$ has (normalized) area at least $\frac12\left(\frac{r}{2}\right)^{p-1}$.
\end{lemma}

\begin{lemma}[Upper bounds on the area of spherical caps]
\label{lem:area-spherical-cap-1}
The spherical cap of height $t$ has (normalized) area at most
\begin{gather*}
\textstyle \left(1-t^2\right)^{p/2}\text{ for any }t\in\left[0,1/\sqrt{2}\right] \\
\textstyle \left(\frac{1}{2t}\right)^p\text{ for any }t\in\left[1/\sqrt{2},1\right).
\end{gather*}
\end{lemma}

We are now ready to analyze the proto-algorithm. Our analysis focuses on the solid angles of normal cones at the extreme points $h_i$ of a convex polytope $P\subset\reals^p$. To simplify notation, we shall say $\omega_i$ in lieu of $\omega(N_P(h_i))$ when the polytope $P$ and extreme point $h_i$ are clear from context. The main result shows we need $O(k\log k)$ optimizations to find all the extreme points with high probability.

\begin{theorem}
\label{thm:exact-recovery}
If $m > \kappa\log\left(\frac{k}{\delta}\right)$, $\kappa = 1 / \log\left(\frac{1}{\max_i1-2\omega_i}\right)$, then the proto-algorithm finds all $k$ extreme points with probability at least $1-\delta$.
\end{theorem}

\begin{proof}
Let $h_i,\,i=1,\dots,k$ be the extreme points. By a union bound,
\BEQ
\Pr\left(\{\text{miss any }h_i\}\right) \le \sum_{i=1}^k\Pr\left(\{\text{miss }h_i\}\right)
\label{eq:miss-ext-point-1}
\EEQ
By the optimality conditions for optimizing a linear function, denoted $g_j,$ over a convex polytope, the event $\{\text{miss }x_i\}$ is equivalent to
\[
\bigcap_{j=1}^m\left\{g_j\notin N_P(h_i)\cup -N_P(h_i)\right\}.
\]
Since the (random) linear functions $g_j,\,j=1,\dots,m$ are \iid{} $\cN(0,I)$, we have
\begin{align*}
\Pr\left(\{\text{miss }h_i\}\right) = \prod_{j=1}^m\Pr\left(g_j\notin N_P(h_i)\cup -N_P(h_i)\right) = (1 - 2\omega_i)^m.
\end{align*}
We substitute this expression into \eqref{eq:miss-ext-point-1} to obtain
\[
\Pr\left(\{\text{miss }h_i\}\right) \le \sum_{i=1}^k(1-2\omega_i)^m \le k\bigl(\max_i1-2\omega_i\bigr)^m.
\]
If we desire the probability of missing an extreme point to be smaller than $\delta$, then we must optimize at least
$$
m > \kappa\log\left(\frac{k}{\delta}\right),\,\kappa = 1 / \log\left(\frac{1}{\max_i1-2\omega_i}\right)
$$
linear functions.
\end{proof}

The constant $\kappa = 1 / \log\left(\frac{1}{\max_i1-2\omega_i}\right)$ is smallest when $\omega_1 = \dots = \omega_k = \frac1k$. Thus, $\kappa$ is at least $1 / \log\left(\frac{1}{\max_i1-2/k}\right),$ which is approximately $\frac{k}{2}$ when $k$ is large. Since $\kappa$ grows linearly with $k$, we restate Theorem \ref{thm:exact-recovery} in terms of the normalized constant
\[
\bar{\kappa} = \frac{1}{k\log\left(\frac{1}{\max_i1-2\omega_i}\right)}.
\]

\begin{corollary}
\label{cor:exact-recovery}
If $m > \bar{\kappa}\,k\log\left(\frac{k}{\delta}\right)$, then the proto-algorithm finds all $k$ extreme points with probability at least $1-\delta$.
\end{corollary}

\subsection{Simplicial constants and solid angles} The constant $\kappa$ is a condition number for the problem. $\kappa$ is large when the smallest normal cone at an extreme point is small. If $\omega_i$ is small, then
\[
\Pr\left(\{\text{miss }h_i\}\right) = (1-2\omega_i)^m
\]
is close to one. Intuitively, this means the polytope has extreme points that protrude subtly. The \emph{simplicial constant} makes this notion precise. For any extreme point $h_i$, the simplicial constant is
\BEQ
\label{eq:simplicial-constant}
\alpha_P(h_i) = \inf_x\left\{\norm{h_i - x}_2\mid x\in\conv\left(\ext(P)\setminus h_i\right)\right\}.
\EEQ
The simplicial constant is simply the distance of the extreme point $h_i$ to the convex hull formed by the other extreme points. To simplify notation, we shall say $\alpha_i$ in lieu of $\alpha_P(h_i)$ when the polytope $P$ and extreme point $h_i$ are clear from context.








The following pair of lemmas justifies our intuition that an extreme point with a small normal cone protrudes subtly and vice versa. We differ the proofs to the appendix.

\begin{lemma}
\label{lem:simplifical-constant-1}
Let $P\subset\reals^k$ be a (convex) polytope and $h_i\in\ext(P)$. If the solid angle of $N_P(h_i)$ is $\omega_i$, then the simplicial constant
\[
\alpha_P(h_i) = \alpha_i \leq R_{\max}\,\frac{r(\omega_i)\sqrt{1-\frac14r(\omega_i)^2}}{1-\frac{1}{2}r(\omega_i)^2},
\]
where $r(\omega) = 2(2\omega)^\frac{1}{k-1}$ and $R_{\max}$ is a constant independent of $h_i.$
\end{lemma}

\begin{remark}
Though we present Lemma \ref{lem:simplifical-constant-1} with a constant dependent on the geometry of the polytope, we observe that this constant is bounded by a quantity that is independent of $h_i$ and depends only on its ``base.'' Such geometric dependence is necessary because $\omega$ is scale invariant while $\alpha_i$ is not. In fact, $\alpha_i$ and $\diam{\left(B_P(h_i)\right)}$ depends on scale in the same manner as $\alpha_i$ does, and thus implicitly adds the appropriate scaling to our bound.
\end{remark}

\begin{lemma}
\label{lem:simplifical-constant-2}
Let $P\subset\reals^k$ be a (convex) polytope and $h_i\in\ext(P)$. If the simplicial constant is $\alpha_P(h_i) = \alpha_i$, then
\[
\omega \left( N_P(h_i) \right) \le \left(\frac{\sqrt{\alpha_i^2 + \left(r_{\min}\right)^2}}{2r_{\min}}\right)^k
\]
when
\[
\frac{\left(r_{\min}\right)^2}{\alpha_i^2 + \left(r_{\min}\right)^2} \ge \frac12,
\]
and where $r_{\min}$ is a constant that depends on geometric properties of the polytope.
\end{lemma}

\begin{remark}
Similar to the situation for $R_{\max},$ $r_{\min}$ depends on geometric properties of the polytope.
\end{remark}

To our knowledge, Lemmas \ref{lem:simplifical-constant-1} and \ref{lem:simplifical-constant-2} are new. The constants $R_{\max}$ in Lemma \ref{lem:simplifical-constant-1} and $r_{\min}$ in \ref{lem:simplifical-constant-2} are non-optimal but their dependence on $P$ is unavoidable since normal cones are scale invariant, but simplicial constants are not. Although sharper bounds on the area of spherical caps are known,\footnote{In fact, exact expressions in terms of the hypergeometric function or the regularized incomplete beta function are known.} we state our results in the aforementioned form for the sake of clarity.

In the literature on NMF, a common assumption is the simplical constant of any extreme point is at least some $\alpha > 0$. By Lemma \ref{lem:simplifical-constant-1}, the simplicial constant being at least $\alpha$ implies
\begin{align*}
\min_i\,\omega_i &\ge \frac12\left( \frac{1-\sin\left(\arctan\left(R_{\max}/\alpha\right)\right)}{2}\right)^{\frac{k-1}{2}} \\
&= \frac12\left( \frac12 - \frac{R_{\max}/\alpha}{2\sqrt{1 + (R_{\max}/\alpha)^2}}\right)^{\frac{k-1}{2}}.
\end{align*}
The relationship between solid angles and simplicial constants is often obscure, and in the rest of the paper, we state results in terms of solid angles $\omega_1,\dots,\omega_k$.

Before we move on to develop variants of the proto-algorithm, we comment on its computational cost in a distributed setting. On distributed computing platforms, communication between the nodes is the major computational cost. Algorithms that make few passes over the data may be substantially faster in practice, even if they require more flops. As we shall see, it is possible to perform NMF or archetypal analysis with just \emph{two passes} over the data.

Consider a typical distributed setting: the data consists of $n$ data points distributed across $D$ nodes of a large cluster. Let $I_d\subset [n]$ be the indicies of the data points stored on the $d$-th node. To perform NMF or archetypal analysis, each node evaluates (random) linear functions on the data points stored locally and returns (i) the indices of the data points that maximize and minimize the linear functions $I_{d,\max},I_{d,\min}\subset I_d$ and (ii) the optimal values. Each node evaluates \emph{the same} set of linear functions on its local data points, so the optimal values are comparable. A node collects the optimal values and finds the maximum and minimum values to find the extreme points. We summarize the distributed proto-algorithm in algorithm \ref{alg:distributed-proto-algorithm}. While we present the algorithm here under the assumption that each node contains a subset of the data points, it is equally amenable to parallelization in the situation where each node contains a subset of the features for all of the data point.

\begin{algorithm}
\caption{Proto-algorithm (distributed)}
\label{alg:distributed-proto-algorithm}
\begin{algorithmic}[1]
\State Choose a random seed and distribute it to all nodes.
\For{$d = 1,\ldots,D$} in parallel
\State Generate a $p\times m$ random matrix $G$ with independent $\cN(0,1)$ entries.
\State Form the product $X_{I_d}G$.
\For{$i = 1,\ldots,n$}
\State Let $\left(V_{i,d,\max},I_{i,d,\max}\right)$ and $\left(V_{i,d,\min},I_{i,d,\min}\right)$ denote pairs of the max and min values in column $i$ of $X_{I_d}G$ and the corresponding index.
\EndFor
\EndFor
\For{$i = 1,\ldots,n$}
\State Let $I_{\max}'$ and $I_{\min}'$ denote the indices of the max and min values in the sets $\left\{ \left(V_{i,d,\max},I_{i,d,\max}\right) \right\}_{i=1}^{d}$ and $\left\{ \left(V_{i,d,\min},I_{i,d,\min}\right) \right\}_{i=1}^{d},$ respectively.
\State Set $I_{\max}\gets I_{\max}\cup I_{\max}',\,I_{\min}\gets I_{\min}\cup I_{\min}'$.
\EndFor
\State Return $H = X_{I_{\max}\cup I_{\min}}$.
\end{algorithmic}
\end{algorithm}

The algorithm makes a \emph{single pass} over the data: each node makes a single pass over its (local) data points to evaluate the linear functions. The subsequent operations are performed on the indices $I_{d,\min},I_{d,\max}$ and optimal values and do not require accessing the data points. The communication costs are also minimal. As long as the nodes are set to produce the same stream of random numbers, the linear functions don't need to be communicated. The only information that must be centrally collected are the pairs of values and indices for the maximum and minimum values in each column of the distributed matrix product.

The proto-algorithm finds the extreme points of the point cloud. We obtain the coefficients $W\in\reals^{n\times k}$ that expresses the data points in terms of the extreme points by solving \eqref{eq:nnls}. The NNLS problem \eqref{eq:nnls} is separable across the rows of $W.$ Thus it suffices to solve $D$ small NNLS problems: each node solves a NNLS problem on  the data points stored locally to obtain the coefficients that represent its (local) data point in terms of the extreme points. Solving the NNLS problem requires a second pass over the data. Thus it is possible to perform archeypal analysis or NMF with two passes over the data.

\subsection{Three practical algorithms}

The proto-algorithm requires the non-negative rank $k$ and the condition number $\kappa$ to be known \emph{a priori} (to set $m$ correctly). In this section, we describe  three practical algorithms: one for noiseless $X$ and two for noisy $X$. When $X$ is noiseless, we seek to recover all the extreme points, no matter how subtly a point protrudes from the point cloud.

\begin{algorithm}
\caption{Noiseless algorithm}
\label{alg:noiseless}
\begin{algorithmic}[1]
\Require $X\in\reals^{n\times p}$
\State Set $I_{\max} = I_{\min} = \emptyset$.
\Repeat
\State Generate a $p\times m$ random matrix $G$ with independent $\cN(0,1)$ entries.
\State Form the product $XG$.
\State Find the indices of the max $I_{\max}'$ and min $I_{\min}'$ in each column of $XG$.
\State Set $I_{\max}\gets I_{\max}\cup I_{\max}',\,I_{\min}\gets I_{\min}\cup I_{\min}'$.
\Until{$I_{\max}',\,I_{\min}'$ adds nothing to $I_{\max}',\,I_{\min}'$.}
\State Return $H = X_{I_{\max}\cup I_{\min}}$.
\end{algorithmic}
\end{algorithm}

The noiseless algorithm stops when $m$ optimization find no missed extreme points ($m$ failures). This stopping rule admits an \emph{a posteriori} estimate of the size of the normal cone at any missed extreme point. Consider each optimization as a Bernoulli trial with $p = 2\sum_{i\in I_\text{miss}}\omega_i$ (success is finding a missed extreme point). The noiseless algorithm stops when we observe $m$ failures.
A $1-\alpha$ confidence interval for $p$ is
\[
\sum_{i\in I_\text{miss}} \omega_i \le \frac{1}{2m}\log\left(\frac{1}{\alpha}\right)\text{ with probability }1-\alpha.
\]

\begin{lemma}
\label{lem:noiseless-algorithm}
The noiseless algorithm finds all extreme points with $\omega_i \ge \frac{1}{2m}\log\left(\frac{1}{\delta}\right)$ with probability at least $1-\delta$.
\end{lemma}

In the presence of noise, we seek to select ``true'' extreme points and discard spurious extreme points created by noise. Since optimizing linear functions over the point cloud gives both true and spurious extreme points, we propose two approaches to selecting extreme points.

The first approach is based on the assumption that spurious extreme points protrude subtly from the point cloud. Thus the normal cones at spurious extreme points are small, and these points are less likely to be found by optimizing linear functions over the point cloud.  This suggests a simple approach to select extreme points: keep the points that are found most often.

The second approach is to select extreme points by sparse regression. Given a set of extreme points (rows of $H$), we solve a group lasso problem (each group corresponds to an extreme point) to select a subset of the points:
\BEQ
\label{eq:group-lasso}
\begin{aligned}
& \minimize_W & & \frac12\norm{X - WH}_F^2 + \lambda\sum_{i=1}^k\norm{w_i}_2 \\
& \subjectto & & W\ge 0.
\end{aligned}
\EEQ
where $\lambda$ is a regularization parameter that trades-off goodness-of-fit and group sparsity. The group lasso was proposed by \cite{yuan2006model} to select groups of variables in (univariate) regression and extended to multivariate regression by \cite{obozinski2011support}. Recently, \cite{kim2012group} propose a similar optimization problem for NMF.

We enforce a non-negativity constraint to keep $W$ non-negative. Although seemingly innocuous, most first-order solvers cannot handle the nonsmooth regularization term and the non-negativity constraint together. Fortunately, a simple reformulation allows us to use off-the-shelf first-order solvers to compute the regularization path of \eqref{eq:group-lasso} efficiently. The reformulation hinges on a key observation.

\begin{lemma}
\label{lem:projection-onto-intersection}
The projection of a point $x\in\reals^p$ onto the intersection of the second-order cone $K_2^p = \{x\in\reals^p\mid\norm{x_{[p-1]}}_2 \le x_p\}$ and the non-nega\-tive orthant $\reals_+^p$ is given by
\[
P_{K_2^p\,\cap\,\reals_+^p}(x) = P_{K_2^p}\left(P_{\reals_+^{p-1}\times\reals}\left(x\right)\right).
\]
\end{lemma}

Although we cannot find Lemma \ref{lem:projection-onto-intersection} in the literature, this result is likely known to experts. For completeness, we provide a proof in the Appendix. We formulate \eqref{eq:group-lasso} as a second-order cone program (SOCP) (with a quadratic objective function):
\[
\begin{aligned}
& \minimize_{W,t} & & \frac12\norm{X - WH}_F^2 + \lambda\sum_{i=1}^k t_i \\
& \subjectto & & \norm{w_i}_2 \le t_i,i = 1,\dots k \\
& & & W\ge 0.
\end{aligned}
\]
Since $t_i,i = 1,\dots k$ are non-negative, the problem is equivalent to
\BEQ
\begin{aligned}
& \minimize_{W,t} & & \frac12\norm{X - WH}_F^2 + \lambda\sum_{i=1}^k t_i \\
& \subjectto & & (w_i,t_i)\in K_2^{n+1}\cap\reals_+^{p+1},i= 1,\dots k.
\end{aligned}
\label{eq:group-lasso-socp-form}
\EEQ
Since the we know how to projection onto the feasible efficiently, most off-the-shelf first-order solvers (with warm-starting) are suited to computing the regularization path of \eqref{eq:group-lasso-socp-form}.

In practice, the non-negative rank $k$ is often unknown. Fortunately, both approaches to selecting extreme points also give estimates for the (non-negative) rank. In the greedy approach, an ``elbow'' on the scree plot of how often each extreme point is found indicates how many extreme points should be selected. In the group lasso approach, persistence of groups on the regularization path indicates which groups correspond to ``important'' extreme points; \ie{} extreme points that are selected by the group lasso on large portions of the regularization path should be selected.

\section{Simulations}
\label{sec:simulations}

We conduct simulations to
\BNUM
\item validate our results on exact recovery by archetype pursuit.
\item evaluate the sensitivity of archetype pursuit to noise.
\ENUM

\subsection{Noiseless}
To validate our results on exact recovery, we form matrices that we know admit a separable NMF and use our algorithm to try and find the matrix $H.$ We construct one example to be what we consider a well conditioned matrix, \ie{} all of the normal cones are large, and we construct another example where the matrix is ill conditioned, \ie{} some of the normal cones may be small.

In order to construct matrices to test the randomized algorithm we use the following procedure. First, we construct a $k \times p$ matrix $H$ and a $n \times k$ matrix $W$ such that $X = WH$ has a separable NMF. The matrix $W$ contains the identity matrix as its top $k \times k$ block and the remainder of its entries are drawn from uniform random variables on $[0,1]$ and then each row is normalized to sum to one. This means that given the matrix $X$ we know that the first $k$ rows of $X$ are the rows we wish to recover using our algorithm.

In Section \ref{sec:archetype-pursuit} we discussed the expected number of random linear functions that have to be used in order to find the desired rows of the matrix $X$ with high probability. To demonstrate these results we use Algorithm \ref{alg:noiseless} with various choices of $m$ and see if the algorithm yields the first $k$ rows of $X.$

To generate the plots shown here we vary $k$ and for each $k$ we vary the number of random projections used, $m.$ For each pair of $k$ and $m$ we construct matrices $W$ and $H$ 500 times, run the algorithm on the resulting $X$ and report the percentage of time that the algorithm correctly found the first $k$ rows of $X$ to be the necessary columns to form a separable NMF. For all of the experiments here we use $n=500$ and $p=1000.$

To demonstrate the algorithm on a well conditioned example we construct the matrix $H$ to have independent entries each of which are uniform on $[0,1].$ We expect the convex hull of the point cloud formed by $H$ to have reasonably sized normal cones. Figure \ref{fig:rand_test} shows the recovery percentages for this experiment as we vary $m$ and $k.$ We measure the number of random linear functions used as a factor times $k.$ To show the scaling that we expect, up to the aforementioned constant, we also plot the line $m/k = \log k.$ Finally, we plot the $95\%$ isocline. We observe that the isoclines behave like $m = k \log k$ and in fact appear to grow slightly slower. Furthermore, in this case the constant factor in the bounds appears to be very small.

\begin{figure}[ht!]
  \centering
  \includegraphics[width = 1\textwidth]{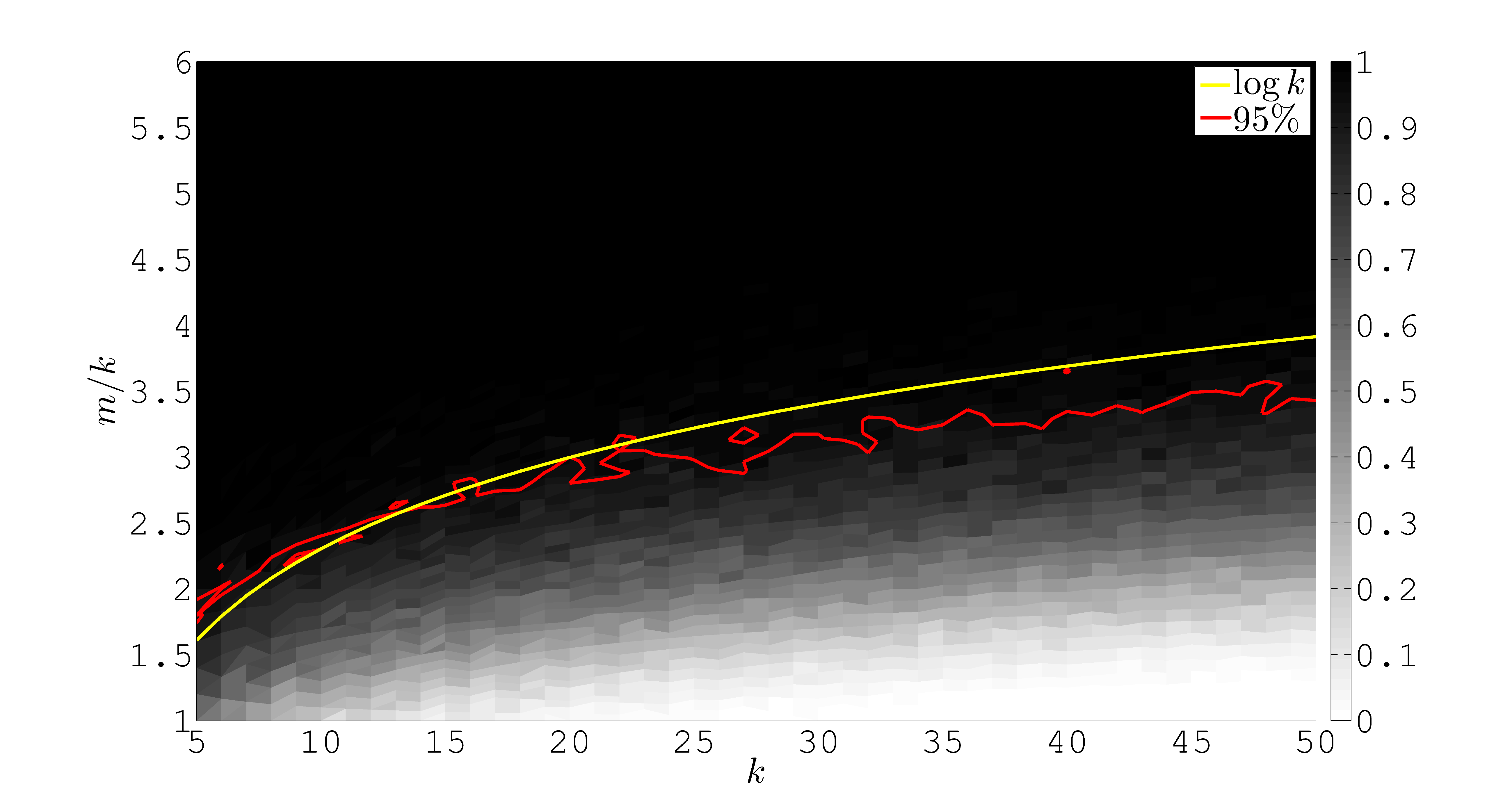}
  \caption{Percentage of experiments in which the algorithm correctly identified the first $k$ rows of $X$ as the rows of $H$ in a separable NMF. Here $H$ is a matrix with independent entries each of which is uniform on $[0,1].$}
  \label{fig:rand_test}
\end{figure}

For our poorly conditioned example we take the matrix $H$ to be the first $k$ rows of the $p \times p$ Hilbert matrix, whose $i,j$ entry is given by $\frac{1}{i+j-1}.$
This matrix is notoriously ill conditioned in the classical sense, \eg\ for a $1000 \times 1000$ Hilbert matrix the computed condition number of a matrix constructed from the first 50 rows is on the order of $10^{17}$ and may in fact be considerably larger. Because even a reasonably small subset of the leading rows of the Hilbert matrix are very close to linearly dependent we expect that the convex polytope defined by there points is very flat and thus some of the extreme points have very small normal cones. Figure \ref{fig:hilb_test} shows the recovery percentages for this experiment as we vary $m$ and $k.$ Similar to before we measure the number of random linear functions used as a factor times $k.$ Once again, to demonstrate the scaling that we expect, up to the aforementioned constant, we also plot the line $m/k = \log k.$ As before, we also plot the $95\%$ isocline.

 We observe that once again the isoclines behave like $m = k \log k,$ though in this case the constant factor is considerably larger than it was before. Given the interpretation of this experiment as trying to find the NMF of a very ill conditioned matrix we expected to observe a larger constant for complete recovery. Though, the algorithm does not require an unreasonable number of projections to recover the desired columns. In fact here we see that in order to recover the correct columns close to $95\%$ of the time we require $m$ to be slightly larger than $10 k \log k.$

\begin{figure}[ht!]
  \centering
  \includegraphics[width = 1\textwidth]{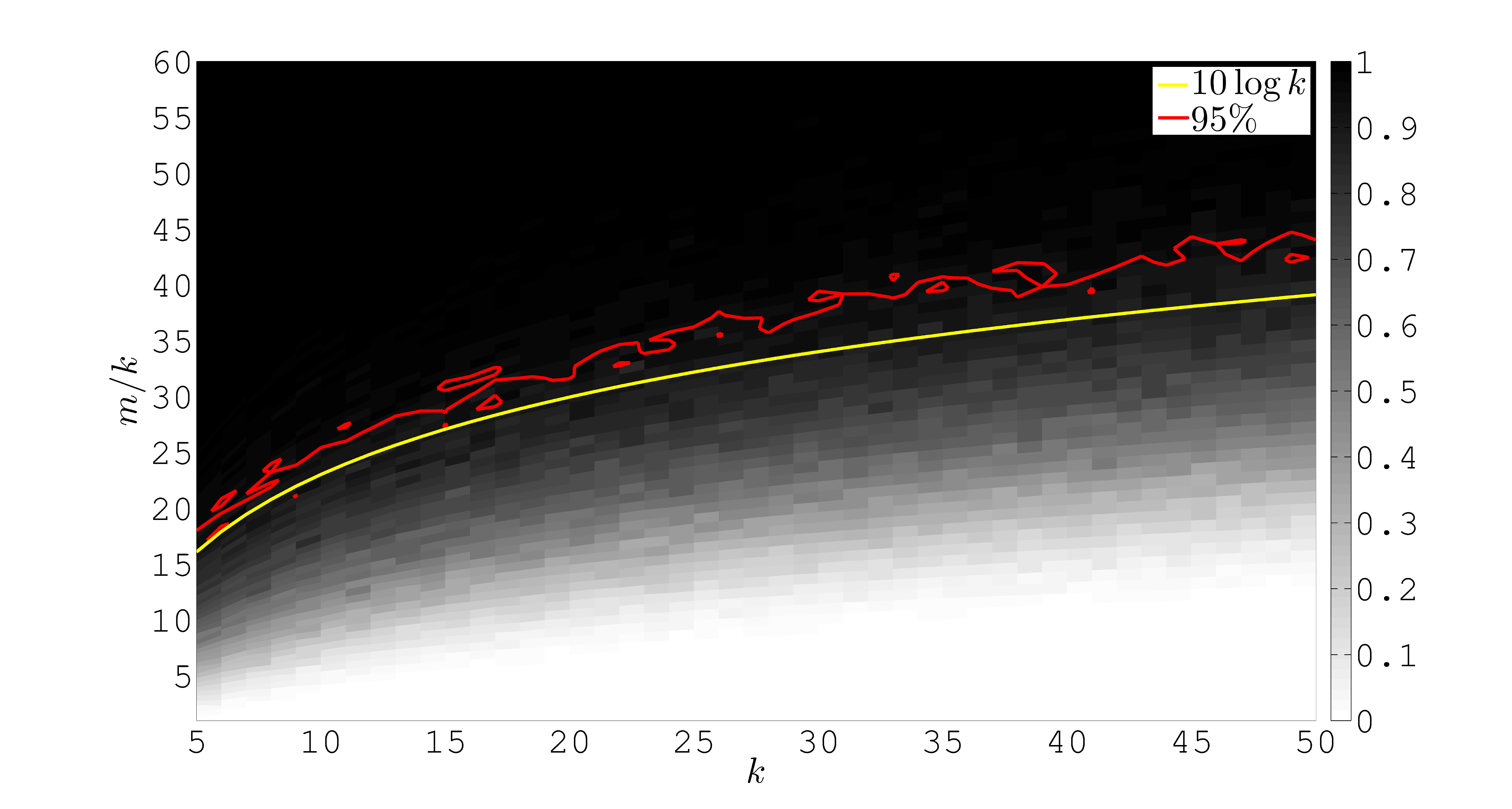}
  \caption{Percentage of experiments in which the algorithm correctly identified the first $k$ rows of $X$ as the rows of $H$ in a separable NMF. Here $H$ is the first $k$ rows of a $p \times p$ Hilbert matrix.}
  \label{fig:hilb_test}
\end{figure}

\subsection{Noisy}

We now demonstrate the performance of the algorithm when rather than being given the matrix $X=WH,$ we instead have a matrix of the form $\tilde{X} = WH + N,$ where $N$ represents additive noise.

For the first example, similar to before, we construct $H$ to be a $20 \times 1000$ matrix. However, now, similar to the experiments in \cite{gillis2012fast} we let $W^T = [I \; W_2]$ where $W_2$ is a $20 \times \binom{20}{2}$ matrix whose columns are all the possible combinations of placing a $1/2$ in two distinct rows and 0 in the remaining rows. Finally, the matrix $N$ is constructed with independent $\epsilon \cN(0,1)$ entries.

To demonstrate the performance of the algorithm on this noisy example we ran the algorithm using the majority voting scheme on matrices with varying levels of noise. We fixed the nonnegative rank to be 20 and took various values of $m$ and $\epsilon.$ For each value of $m$ and $\epsilon$ we constructed the matrices $W,$ $H,$ and $N$ as previously described. After forming the the matrix $X_N$ we ran the algorithm 50 times on the matrix. Each time the 20 most frequently found rows are collected into the rows of a matrix denoted $\tilde{H}$ and the rows of $\tilde{W}$ are computed using nonnegative least squares to try and satisfy $X=\tilde{W}\tilde{H}.$

Figure \ref{fig:noise} shows the error computed as
\[
\frac{\| X - \tilde{W}\tilde{H} \|_F}{n}
\]
for various values of $m$ and $\epsilon$ on a $\log_{10}$ scale. Each pixel represents the average error over 50 trials. We observe that as expected the overall error increases with $\epsilon,$ but that after an appropriate number of random projections the error does not significantly decay.

To complement the plot of the residual error, we demonstrate the behavior of the random voting scheme itself in the presence of noise. To do this, we construct 100 matrices $X_N$ for 10 distinct $\epsilon$ and use $m = 20 k \log{k}.$ Figure \ref{fig:scree-test} shows a sorted version of the number of times each row is found, as a fraction of the maximum number of votes a singe row received. Each row of the image represents an experiment, and each block of 100 rows corresponds to a fixed noise level. As expected we see that there is a significant drop off in votes between the 20 significant rows and the remaining columns as long as the noise is small. Once the noise becomes larger, we see that more points are becoming relevant extreme points and thus there is no longer a sharp transition at 20. One interesting note is that, because each row receives at least one vote, adding the noise has perturbed the convex polytope in a way such that all points are now extreme points.

\begin{figure}[ht!]
  \centering
  \includegraphics[width = 1\textwidth]{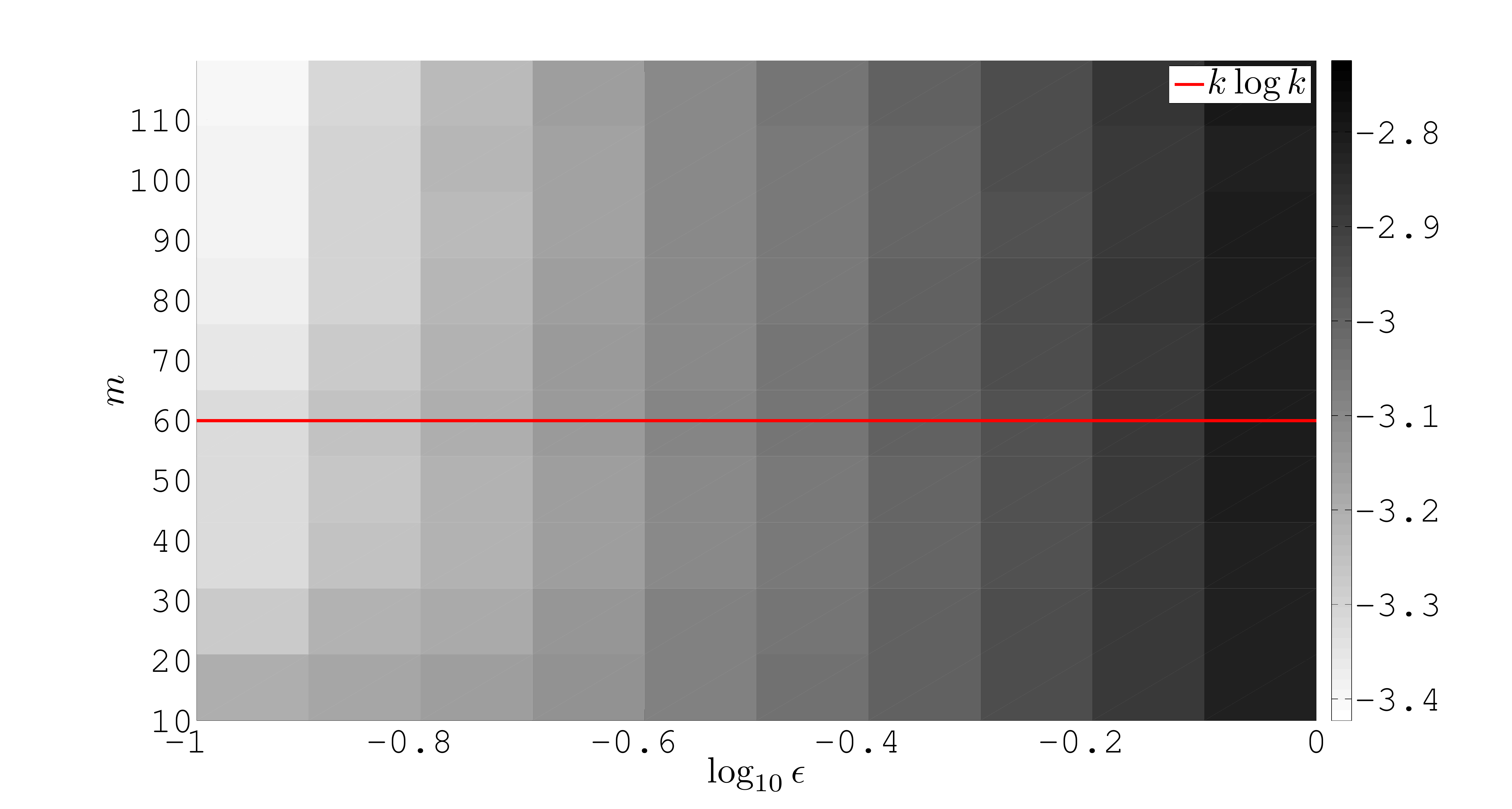}
  \caption{$\| X - \tilde{W}\tilde{H} \|_F / n$ on a $\log_{10}$ scale for various $\epsilon$ and $m$ when using a majority vote scheme to select $\tilde{H}.$}
  \label{fig:noise}
\end{figure}

\begin{figure}[ht!]
  \centering
  \includegraphics[width = 1\textwidth]{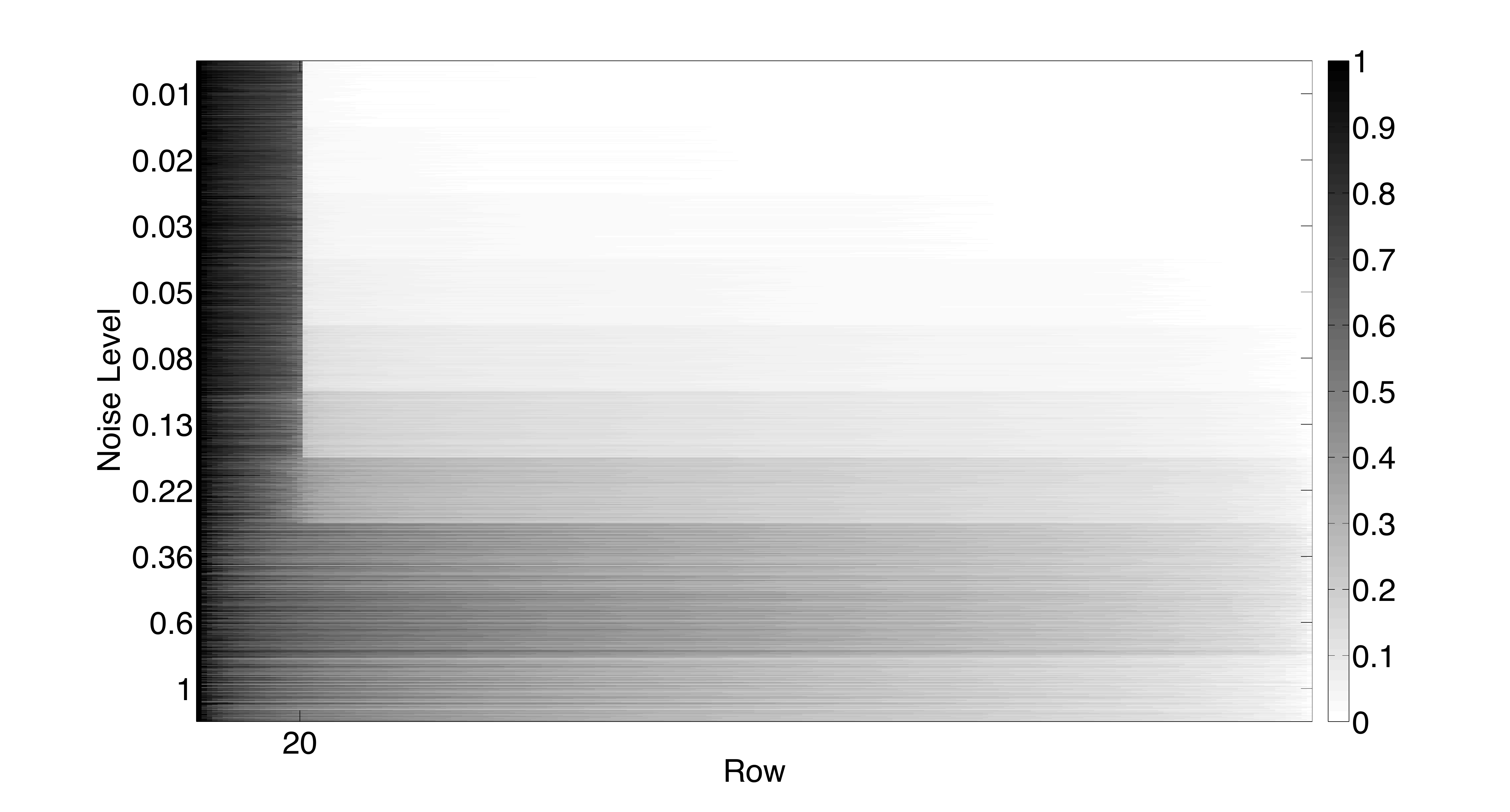}
  \caption{Number of votes received for each row, sorted and normalized by the largest number of votes received. Each row of the image represents a distinct instance of the experiment, and each block of 100 rows corresponds to a given noise level, $\epsilon$.}
  \label{fig:scree-test}
\end{figure}

Finally, the demonstrate the behavior of the algorithm when coupled with the group LASSO approach for picking rows we ran the algorithm using the same setup as for the random voting example. This means that we fixed the rank at 20 and used the group LASSO path to pick which 20 columns, of those found via the prototype algorithm, should form $\tilde{H}.$ The rows of $\tilde{W}$ are then computed using nonnegative least squares to try and satisfy $X=\tilde{W}\tilde{H}$.

Figure \ref{fig:LASSO_residual} shows the error computed as before for various values of $m$ and $\epsilon$ on a $\log_{10}$ scale. Each pixel represents the average error over 50 trials. We observe, once again, that as expected the overall error increases with $\epsilon,$ but that the algorithm is not sensitive to the number of random linear functions used. Even a small number of random linear functions is sufficient to identity the key columns.

\begin{figure}[ht!]
  \centering
  \includegraphics[width = 1\textwidth]{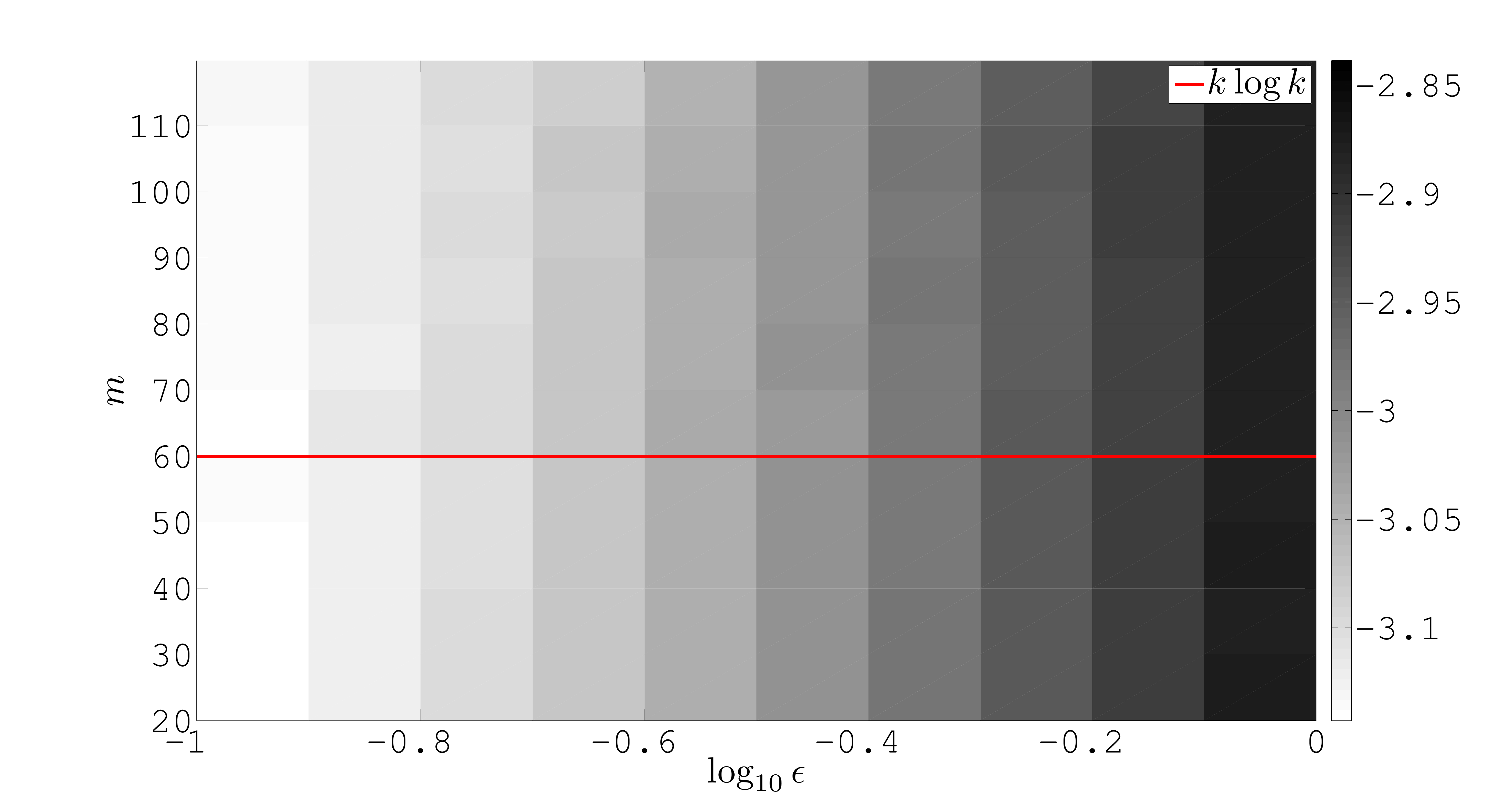}
  \caption{$\| X - \tilde{W}\tilde{H} \|_F / n$ on a $\log_{10}$ scale for various $\epsilon$ and $m$ when using the group LASSO to select the rows of $\tilde{H}.$}
  \label{fig:LASSO_residual}
\end{figure}

\section{Hyperspectral image example}
Based on the origins of PPI, we demonstrate the use of random projections for finding important pixels in a hyperspectral image. We used a hyperspectral image of the National Mall in Washington, DC \cite{landgrebe2003signal}\footnote{Available on the web at \url{engineering.purdue.edu/~biehl/MultiSpec/hyperspectral.html}}. The image is $1280 \times 307$ pixels in size, contains 191 bands and is displayed in Figure . We utilize algorithm \ref{alg:proto-algorithm} to find the important pixels in the image. Intuitively, we should find pixels that represent pure versions of each class of objects, \eg, trees, roofs, roads, etc., in the image. We then use these important pixels to broadly classify the remaining pixels in the image as each type of object. The assumption that predicate such a process is that in the image there appear to be a few key, or dominant, object classes. Figure \ref{fig:mall_scree} shows the relative frequency with which each selected extreme point is chosen. Here, we observe that there are roughly 10-15 key pixels identified by the algorithm.

\begin{figure}[ht!]
  \centering
  \includegraphics[height = \columnwidth, angle =90]{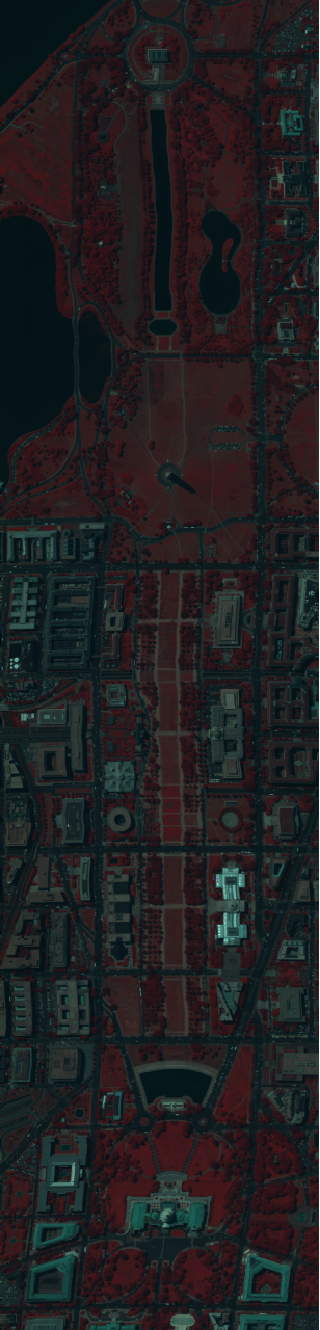}
  \caption{Hyperspectral image of the National Mall in Washington, DC. The RGB values of the image are set by choosing, for each color, a single color band.}
  \label{fig:mall_orig}
\end{figure}

\begin{figure}[ht!]
  \centering
  \includegraphics[width = .7\columnwidth]{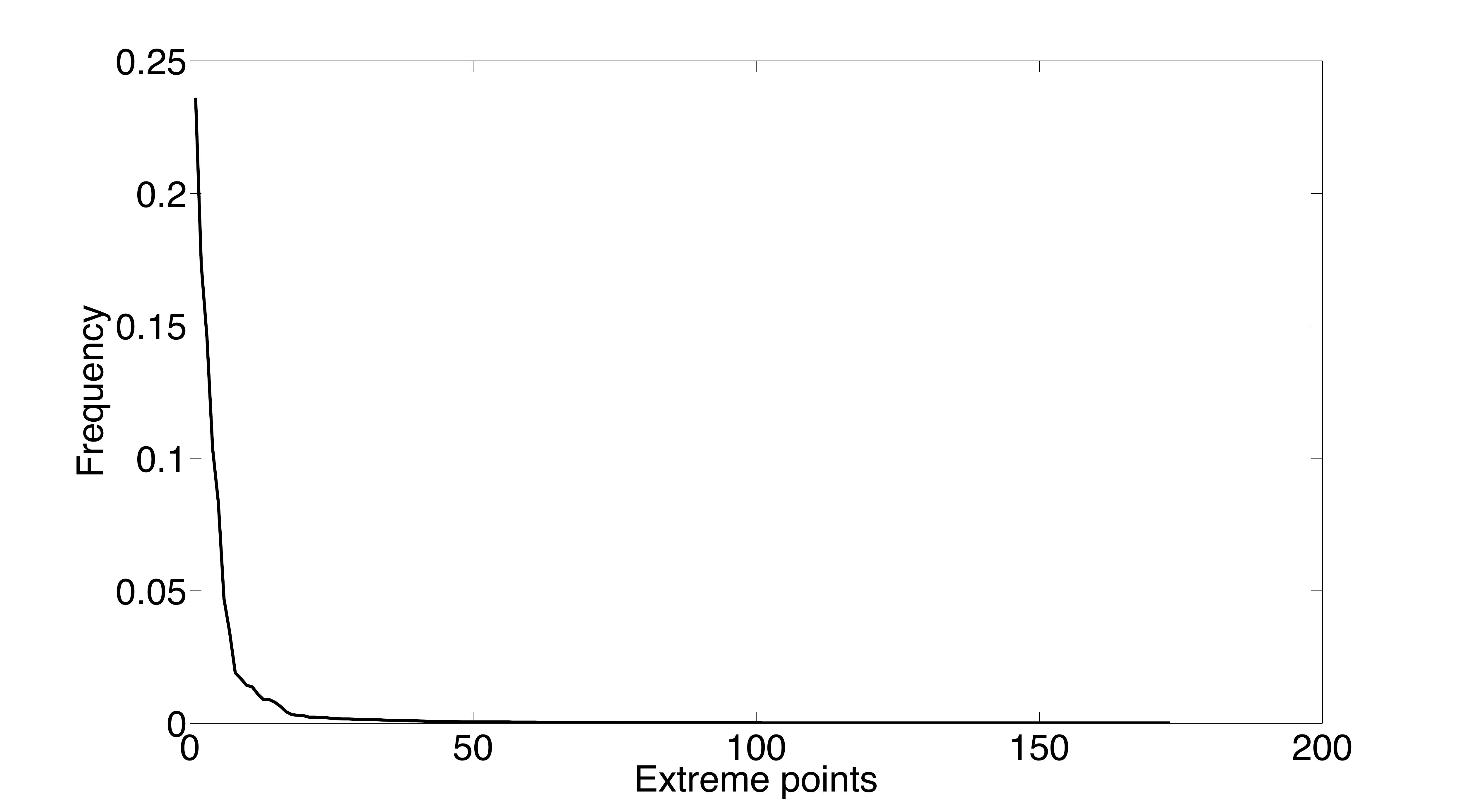}
  \caption{Relative frequency that each extreme point is selected via algorithm \ref{alg:proto-algorithm} using 5,000 random projections.}
  \label{fig:mall_scree}
\end{figure}

We now partially decompose the image using the most interpretable of the 11 most frequently selected extreme points. There are some pixels in the image that may be considered outliers, and because they are distinct from the remaining pixels they will be selected a lot. In fact, these points correspond to very pointy extreme points. These points are, in fact, important as they represent objects unlike the remainder of the image. In this situation, one example is that there appears to be a bright red light on the roof of the National Gallery of Art; such an object does not appear elsewhere in the image. However, for presentation purposes we stick to the important pixels that represent large sections of the image.

To broadly classify the image, we selected four pixels that appear to represent key features. We classify the remaining pixels by simply asking which representative pixel, of the 11 most selected, their spectrum looks most similar to in the $\ell_2$ sense. Figure \ref{fig:mall_class} shows the pixels classified into 4 categories. In each image the pixels that are classified as such are left colored and the remaining pixels are colored black. In fact, the images corresponding to the other seven pixels represent very little of the image.

\begin{figure}[ht!]
    \centering
    \subfloat[]{\includegraphics[width = .2\columnwidth]{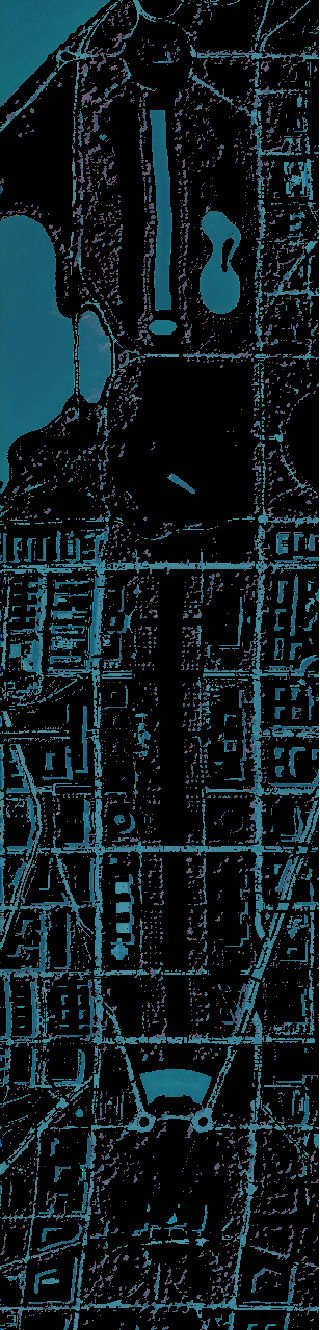}}
    \quad
    \subfloat[]{\includegraphics[width = .2\columnwidth]{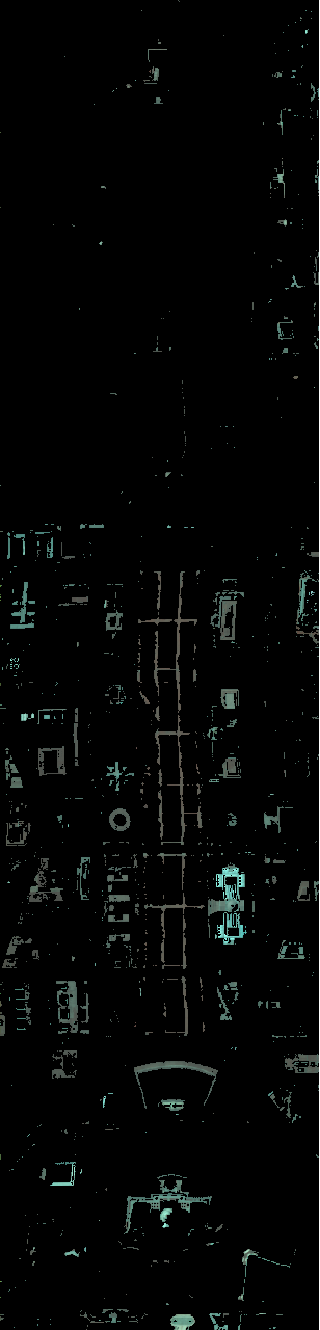}}
    \quad
    \subfloat[]{\includegraphics[width = .2\columnwidth]{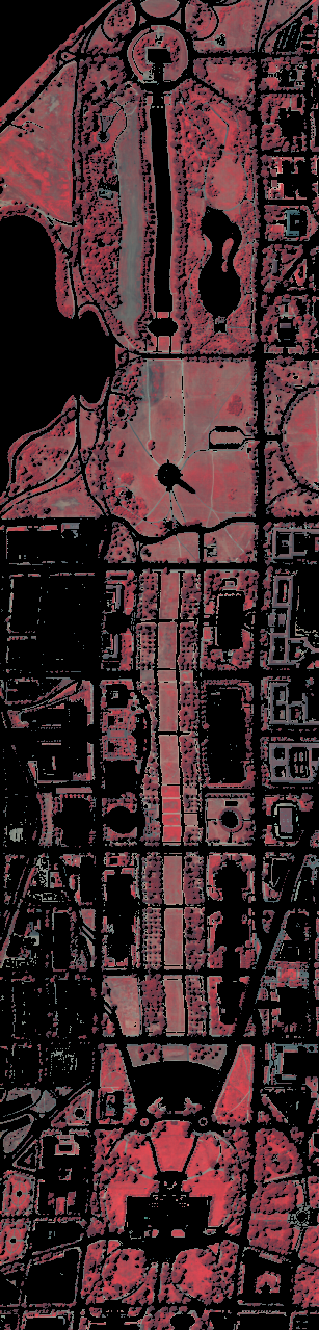}}
    \quad
    \subfloat[]{\includegraphics[width = .2\columnwidth]{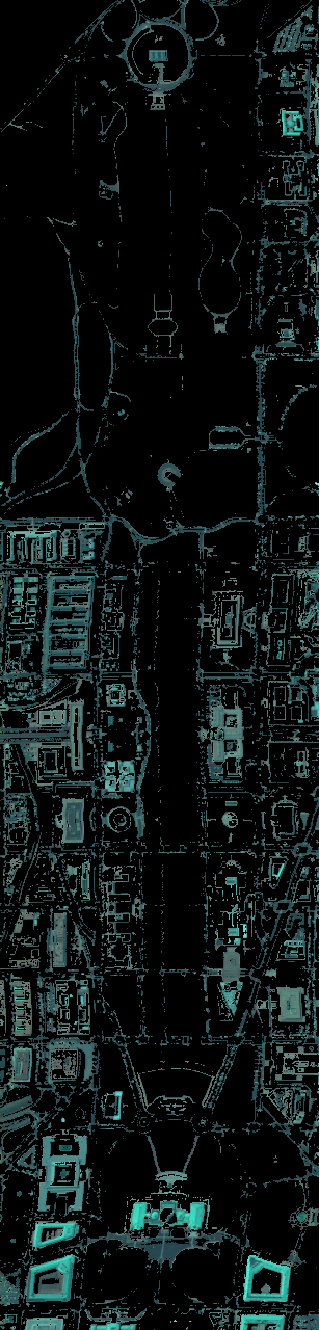}}
    \caption{Classification of pixels into (a) roads and water, (b) concrete, (c) trees and grass, and (d) roofs. The contrast has been exaggerated for presentation.}
   \label{fig:mall_class}
\end{figure}

\section{Hereditary breast cancer dataset}
We adopt the approach of \cite{brunet2004metagenes} to discover ``metagenes'' from gene expression data with NMF. Given a dataset consisting of the expression levels of $d$ genes in $n$ samples, we seek to represent the expression pattern of the samples in terms of conical combinations of a small number of metagenes. The data is usually represented by an expression matrix $X\in\reals^{d\times n}$. In most studies, $d \gg n$. Thus expression matrices are usually ``tall and skinny.'' Mathematically, we seek an approximate factorization of the expression matrix $X = UV^T$ in terms of non-negative factors $U\in\reals^{d\times k}$ and $V\in\reals^{n\times k}$: $X \sim UV^T$. The columns of $U$ are metagenes, and the rows of $V$ are the coefficients of the conical combinations.

The hereditary breast cancer data\-set collected by \cite{Hedenfalk2001Gene} consists of the expression levels of 3226 genes on 22 samples from breast cancer patients. The patients consist of three groups: 7 patients with a BRCA1 mutation, 8 samples with a BRCA2 mutation,
and 7 additional patients with sporadic (either estrogen-receptor-negative, aggressive cancers or estrogen-receptor-positive, less aggressive) cancers. The dataset is available at \url{http://www.expression.washington.edu/publications/kayee/bma/}. We exponentiate the data to make the log-expression levels non-negative.

We normalize the expression profiles (columns of $X$) and look for extreme points with the proto-algorithm. Figure \ref{fig:brca-scree-plot} shows a scree plot of how often each extreme point is found by the proto-algorithm. Figure \ref{fig:brca-scree-plot} also shows a plot of the relative residual versus how many extreme points are selected. The extreme points were selected by keeping the points found most often. On both plots, we notice an ``elbow'' at 6. This suggests the expression matrix is nearly-separable and has non-negative rank 6. Biologically, this means the expression pattern is mostly explained by the expression pattern of 6 metagenes.

\begin{figure} 
\includegraphics[width = 0.48\textwidth]{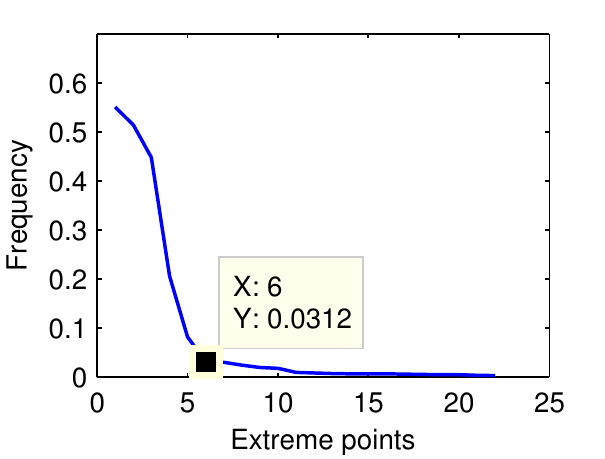}
\includegraphics[width = 0.48\textwidth]{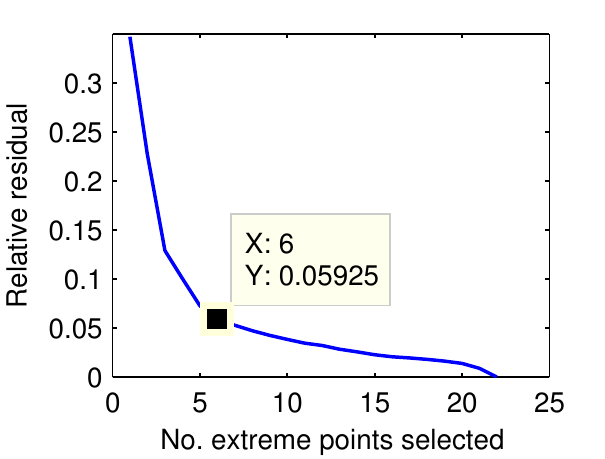}
\caption{Scree plot (left) of how often each extreme point is found and plot of the relative residual $\norm{X - UV^T}_F / \norm{X}_F$ (right) versus how many extreme points are selected. There is a noticeable ``elbow'' at 6 on both plots.} 
\label{fig:brca-scree-plot}
\end{figure}

We also selected metagenes by sparse regression \eqref{eq:group-lasso}. To compute the regularization path of \eqref{eq:group-lasso}, we implemented a solver on top of TFOCS by \cite{becker2011templates}. Figure \ref{fig:brca-group-lasso} shows a coefficient plot and a spy plot of the regularization path. Although the sparse regression approach accounts for correlation among metagenes, the effect is negligible for the beginning (large regularization parameter) of the regularization path. Figure \ref{fig:brca-metagenes} shows the first 4 metagenes selected by the group lasso approach and by the greedy approach are the same and the sixth metagene selected by the greedy approach is the seventh to enter the regularization path.

\begin{figure}
\hspace*{-8pt}\includegraphics[trim=0pc 4pc 0pc 59pt, clip, scale = 1.04]{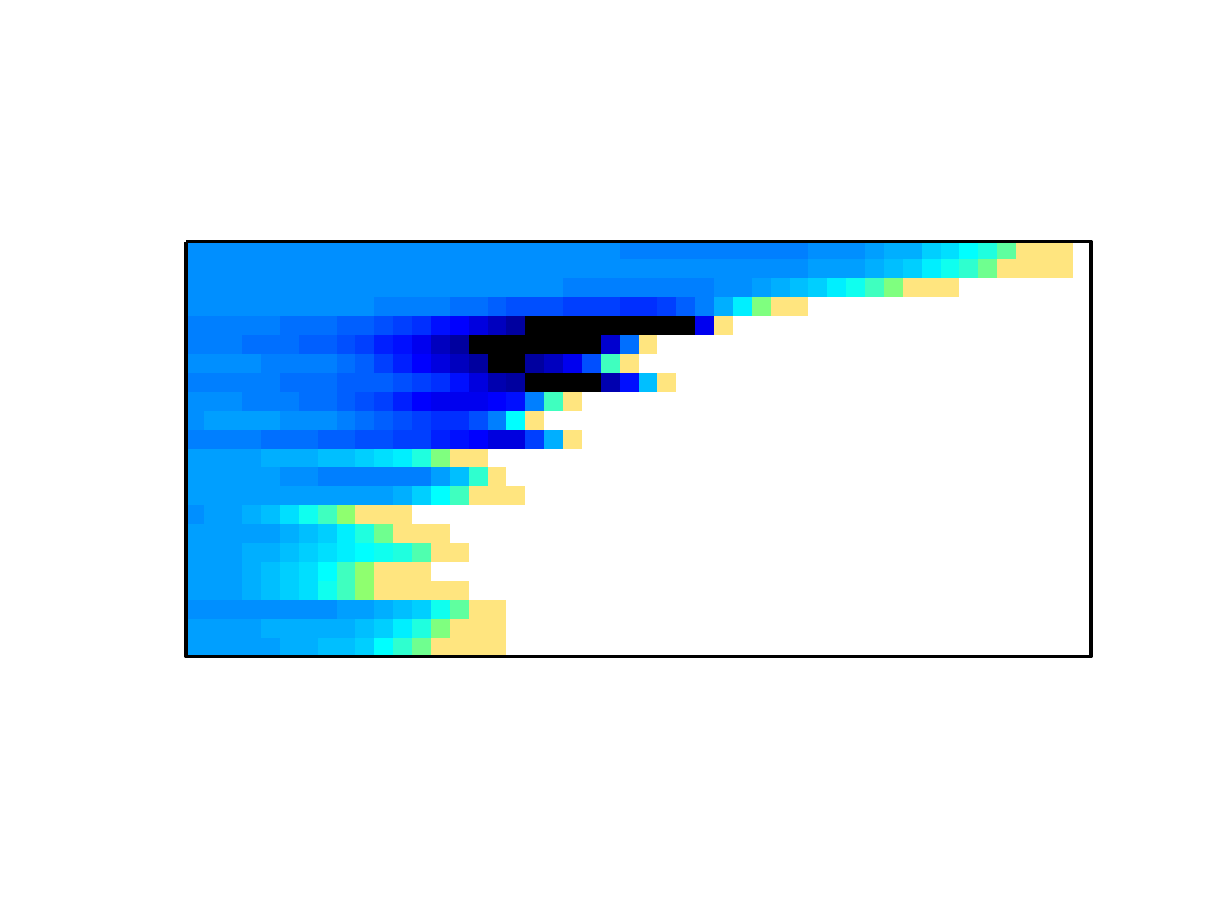} \\
\includegraphics{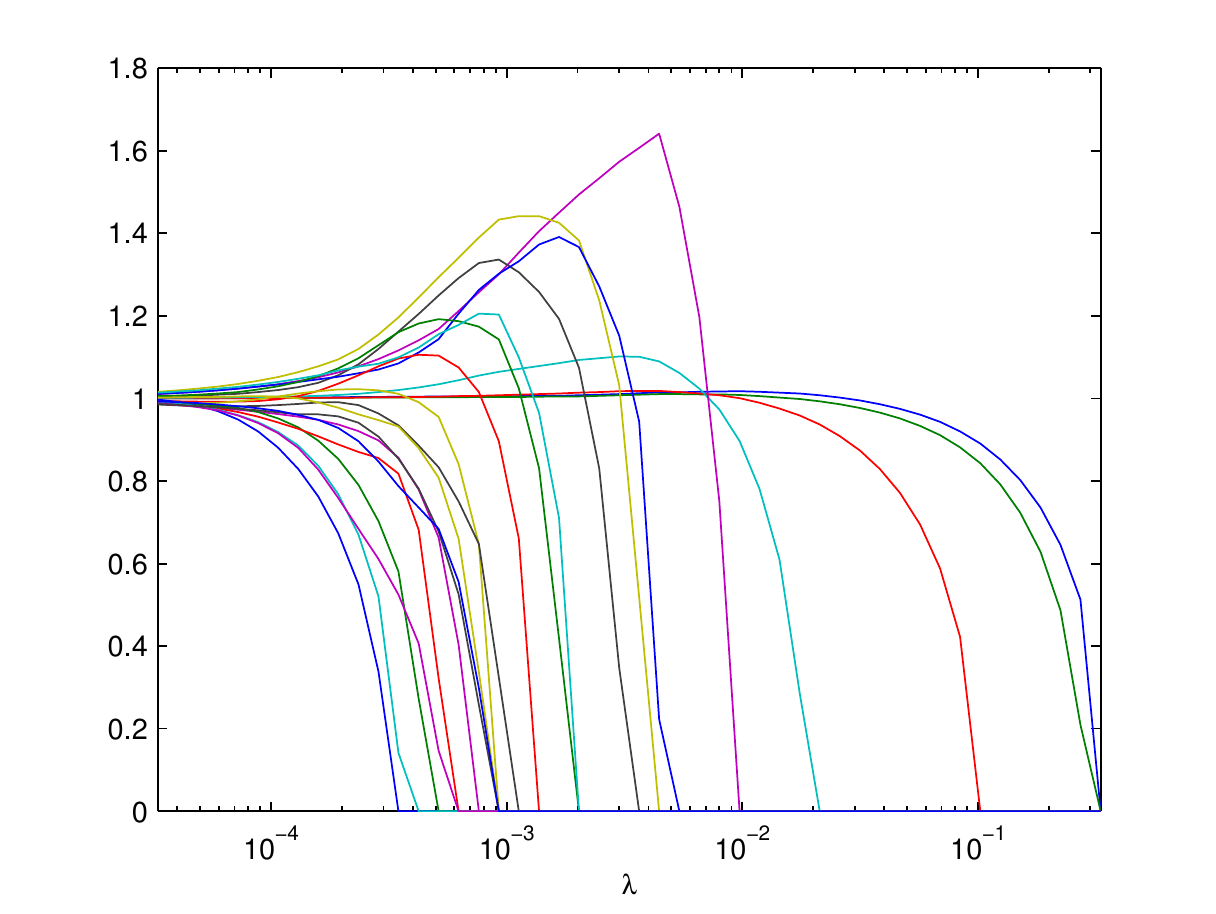}
\caption{A spy plot (top) and a profile (bottom) of the group lasso path ($\norm{v_i}_2,\,i=1,\dots,22$ versus the regularization parameter $\lambda$). The rows (of pixels) in the spy plot and the lines in the profile correspond to groups. In the spy plot, lighter pixels correspond to small coefficients, while darker pixels correspond to large coefficients. }
\label{fig:brca-group-lasso}
\end{figure}

\begin{figure}
\includegraphics{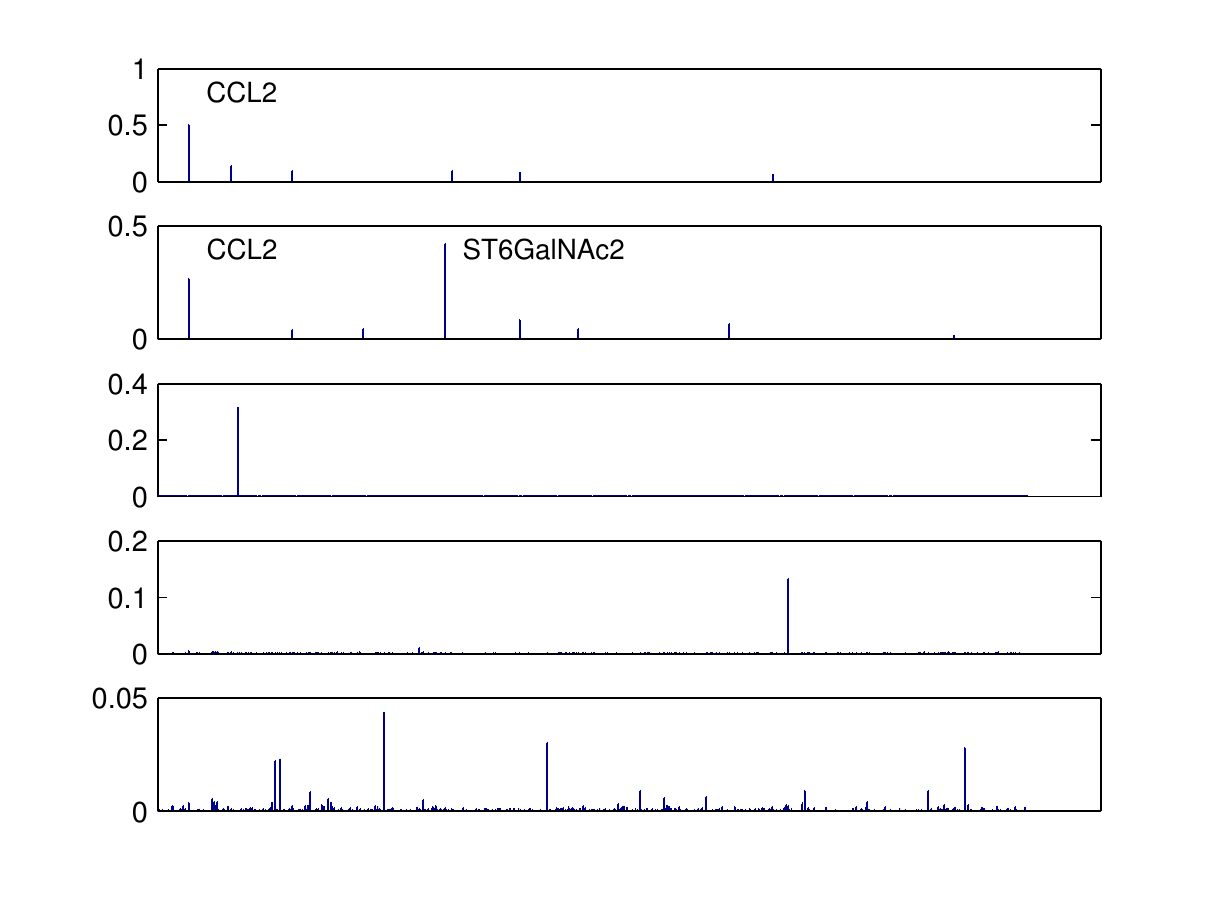}
\caption{The first five metagenes selected by both the greedy and the group lasso approachs. The first two metagenes show high expression levels of an inflammatory chemokine CCL2 (also called MCP-1). \cite{soria2008inflammatory} showed levated CCL2 expression is associated with advanced disease course and with progression in breast cancers. This is consistent with the fact that 12 (of 22) samples in the study were (histologically) graded and all showed moderate to poor-differentiation (grades 6 to 9 on a scale of 1 to 9), an indication of advanced disease progression. The second metagene also shows high expression levels of ST6GalNAc2. Recently, \cite{murugaesu2014in} showed the enzyme
encoded by ST6GalNAc2 is a metastasis suppressor in breast cancers. Unfortunately, the study only included patients with primary cancers, so the data cannot support the association between high expression of ST6GalNAc2 and lower incidence of metastasis. }
\label{fig:brca-metagenes}
\end{figure}

\section{Conclusion and discussion}
\label{sec:conclusion}

Archetype pursuit is a unified approach to archetypal analysis and non-negative matrix factorization. The approach is motivated by a common geometric interpretation of archetypal analysis and separable NMF.
Two key benefits of the approach are
\BNUM
\item \emph{scalability:} The main computational bottleneck is forming the product $XG$, and matrix-vector multiplication is readily parallelizable.
\item \emph{simplicity:} The proto-algorithm is easy to implement and diagnose (when it gives unexpected results).
\ENUM
Furthermore, our simulation results show the approach is robust to noise.

In the context of NMF, an additional benefit is that our approach gives interpretable results even when the matrix is not separable. When the matrix is not separable, the approach no longer gives the (smallest non-negative rank) NMF. However, the geometric interpretation remains valid. Thus, instead of the (minimum non-negative rank) NMF, our approach gives two non-negative factors $W$ and $H$ such that $X\approx WH$, where the rows of $H$ are the extreme rays of a polyhedral cone that contains most of the rows of $X$. An alternative approach in the non separable case is to utilize semidefinite preconditioning techniques proposed by \cite{gillis2013semidefinite}.

\bigskip
\begin{center}
{\large\bf SUPPLEMENTARY MATERIAL}
\end{center}

\begin{proof}[Proof of Lemma \ref{lem:simplifical-constant-1}]
Without loss of generality, assume $h_i$ is the origin. Let $K\subset\reals^k$ be the smallest circular cone with axis
\[
a = \frac{1}{\alpha_i}P_{\conv\left(\ext(P)\setminus h_i\right)}(h_i)
\]
that contains $P$. Since $K$ is a cone, it also contains $T_P(h_i) = \cone(P)$. Thus $K^\circ\subset N_P(h_i)$ and $\omega(K^{\circ}) \le \omega_i$. Figure~\ref{fig:extcone} gives an illustrates the described geometry in a simple case.

\begin{figure}[ht!]
  \centering
  \includegraphics[width = .5\textwidth]{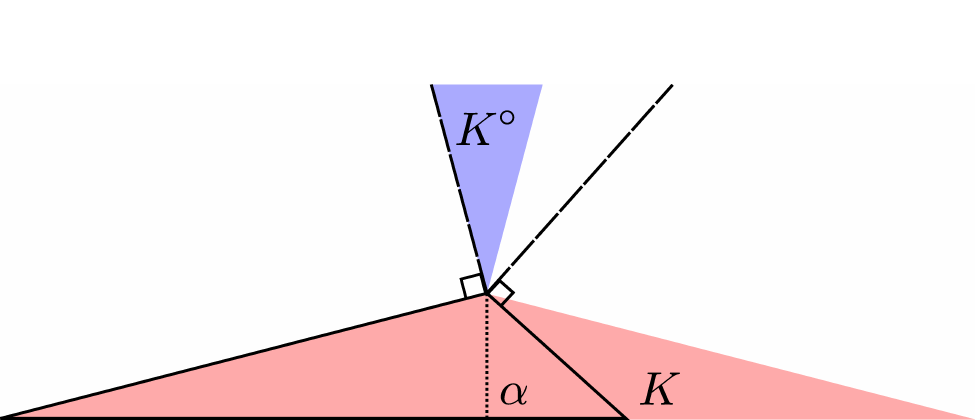}
  \caption{Example of the cone $K$ from the proof of lemma~\ref{lem:simplifical-constant-1}. The normal cone is outlined by the dashed black line and the code $K$ and its polar are shaded.}
  \label{fig:extcone}
\end{figure}

Further, $K^\circ$ is a circular cone with axis $-a$. By Lemma \ref{lem:area-spherical-cap-2}, the radius of the spherical cap $K^\circ\cap\bS^{k-1}$ is at most
\BEQ
r(\omega_i) = 2(2\omega)^{1/(k-1)}.
\label{eq:simplifical-constant-1-1}
\EEQ
Thus the angle of $K^\circ$ is at most $\arccos\left(1-\frac12r(\omega_i)^2\right)$ and the angle of $K$ is at least
\[
\frac{\pi}{2} - \arccos\left(1-\frac12r(\omega_i)^2\right) = \arcsin\left(1-\frac12r(\omega_i)^2\right).
\]
To obtain a bound on the simplicial constant $\alpha_i = \alpha_P(h_i)$, we  study 2-dimensional slices of $K$ and $P$:
\[
K\cap\linspan\left(a,\hat{n}\right)\text{ and }P\cap\linspan\left(a,\hat{n}\right)\text{ for any }\hat{n}\perp a.
\]
Given a slice of $P$ along the direction $\hat{n}$, the simplicial constant is given by
\[
\alpha_i = \frac{r_{\hat{n}}}{\tan(\theta_{\hat{n}})}
\]
for some radius $r_{\hat{n}}$ and some angle $\theta_{\hat{n}}\in\left[0,\frac{\pi}{2}\right)$. Since $K$ is the smallest circular cone (with axis $a$) that contains $P$, the angle for $K$ is equal to $\theta_{\hat{n}}$ for some slice. Further, $P\subset K$ so $r_{\hat{n}}$ is at most the diameter of the ``base'' of the pyramid that is the convex hull of the neighbors of $h_i.$ Mathematically, the base is the set
\[
B_P(h_i) = \conv\left(\left\{u\in\ext(P)\setminus h_i\mid u\text{ shares a face with }h_i\right\}\right).
\]
Thus
\[
\alpha_i \le \frac{\diam{\left(B_P(h_i)\right)}}{\tan\left(\arcsin\left(1-\frac12r(\omega_i)^2\right)\right)} \leq \diam{\left(B_P(h_i)\right)}\frac{r(\omega_i)\sqrt{1-\frac14r(\omega_i)^2}}{1-\frac{1}{2}r(\omega_i)^2}.
\]
\end{proof}

\begin{proof}[Proof of Lemma \ref{lem:simplifical-constant-2}]
The proof is similar to the proof of Lemma \ref{lem:simplifical-constant-1}. Assume (w.l.o.g.) $h_i$ is at the origin. Let $K\subset\reals^k$ be the largest circular cone that sits in $T_P(h_i),$ and let $a \in \bS^{k-1}$ be its axis. Figure \ref{fig:intcone} gives an illustrates the described geometry in a simple case.

\begin{figure}[ht!]
  \centering
  \includegraphics[width = .5\columnwidth]{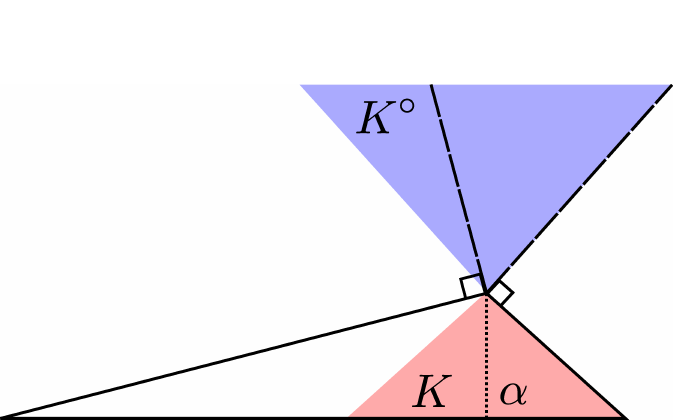}
  \caption{Example of the cone $K$ from the proof of lemma~\ref{lem:simplifical-constant-2}. The normal cone is outlined by the dashed black line and the code $K$ and its are shaded.}
  \label{fig:intcone}
\end{figure}

Let
\[
B = T_P(h_i)\cap\{x\in\reals^k\mid a^T(x - P_{\conv\left(\ext(P)\setminus h_i\right)}(h_i)) = 0\}.
\]
Consider the 2-dimensional slices of $K$ and $P$ given by
\[
K\cap\linspan\left(a,\hat{n}\right)\text{ and }P\cap\linspan\left(a,\hat{n}\right)\text{ for any }\hat{n}\perp a.
\]
Given a slice of $P$ along the direction $\hat{n}$, a bound on the simplicial constant is
\[
\alpha_i \geq \frac{r_{\hat{n}}}{\tan(\theta_{\hat{n}})}
\]
for some radius $r_{\hat{n}}$ and some angle $\theta_{\hat{n}}\in\left[0,\frac{\pi}{2}\right)$. Since $K$ is the largest circular cone that sits in $T_P(h_i)$, the angle of $K$ is equal to $\theta_{\hat{n}}$ for some slice. Further, $r_{\hat{n}}$ is well defined and its value depends only on the geometry of $B.$ We let $r_{\min}$ be the smallest possible value of $r_{\hat{n}}.$ Thus the angle of $K$ is at least $\arctan\left(\frac{r_{\min}}{\alpha_i}\right)$. Since $K^\circ$ is a circular cone with axis $-a$, the angle of $K^\circ$ is at most
\[
\frac{\pi}{2} - \arctan\left(\frac{r_{\min}}{\alpha_i}\right) = \arccot\left(\frac{r_{\min}}{\alpha_i}\right).
\]
An elementary trigonometric calculation shows the height of the spherical cap $\bS^{k-1}\cap K^\circ$ is at least
\[
\cos\left(\arccot\left(\frac{r_{\min}}{\alpha_i}\right)\right) = \frac{r_{\min}}{\sqrt{\alpha_i^2 + \left(r_{\min}\right)^2}}
\]
By Lemma \ref{lem:area-spherical-cap-1}, the solid angle of $K^\circ$ is at most
\[
\left(\frac{\sqrt{\alpha_i^2 + \left(r_{\min}\right)^2}}{2r_{\min}}\right)^k,
\]
since $K\subset T_P(h_i)$, $N_P(h_i)\subset K^\circ$ and $\omega_i \le \omega(K^\circ)$.
\end{proof}


\begin{proof}[Proof of Lemma \ref{lem:projection-onto-intersection}]
Given a closed convex cone $K\subset\reals^n$, a point $x\in\reals^n$ has a unique orthogonal decomposition into $P_K(x) + P_{K^\circ}(x)$. To show $P_{K_2^n}\left(P_{\reals_+^{n-1}\times\reals}\left(x\right)\right)$ is the projection of $x$ onto $K_2^n\cap\reals_+^n$, it suffices to check
\BNUM
\item $P_{K_2^n}\left(P_{\reals_+^{n-1}\times\reals}\left(x\right)\right)\in K_2^n\cap\reals_+^n$
\item $x - P_{K_2^n}\left(P_{\reals_+^{n-1}\times\reals}\left(x\right)\right) \in \left(K_2^n\cap\reals_+^n\right)^\circ = \conv\left(-K_2^n\cap-\reals_+^n\right)$
\item $P_{K_2^n}\left(P_{\reals_+^{n-1}\times\reals}\left(x\right)\right)\perp x - P_{K_2^n}\left(P_{\reals_+^{n-1}\times\reals}\left(x\right)\right)$
\ENUM
for any point $x\in\reals^n$. To begin, we decompose $x$ into its projection onto $\reals_+^{n-1}\times\reals$ and $\left(\reals_+^{n-1}\times\reals\right)^\circ$:
\[
x = P_{\reals_+^{n-1}\times\reals}(x) + P_{\left(\reals_+^{n-1}\times\reals\right)^\circ}(x).
\]
We further decompose $P_{\reals_+^{n-1}\times\reals}(x)$ into its projection onto $K_2^n$ and $K^\circ = -K_2^n$:
\[
P_{\reals_+^{n-1}\times\reals}(x) = P_{K_2^n}\left(P_{\reals_+^{n-1}\times\reals}(x)\right) + P_{-K_2^n}\left(P_{\reals_+^{n-1}\times\reals}(x)\right).
\]
The projection onto $K_2^n$ preserves the zero pattern of $P_{\reals_+^{n-1}\times\reals}(x)$. Thus a point $x\in\reals^n$ admits the decomposition
\[
x = P_{K_2^n}\left(P_{\reals_+^{n-1}\times\reals}(x)\right) + P_{-K_2^n}\left(P_{\reals_+^{n-1}\times\reals}(x)\right) + P_{\left(\reals_+^{n-1}\times\reals\right)^\circ}(x),
\]
where the three parts are mutually orthogonal. Given this decomposition, it is easy to check 1, 2, and 3.
\end{proof}








\bibliographystyle{plain}
\bibliography{yuekai}

\end{document}